\documentclass[10pt,journal]{IEEEtran}

\usepackage[T1]{fontenc}
\usepackage{graphicx}%
\usepackage{color}
\usepackage{multirow}
\usepackage{amsmath,amssymb,amsfonts,amsthm,dsfont}
\usepackage{pdfpages}
\usepackage{arydshln}
\usepackage{fancybox}
\usepackage{enumerate}
\usepackage{url}

\newcommand{\R}{\mathbb R}

\newcommand{\cd}[2]{\frac{\partial #1}{\partial #2}}
\newcommand{\dd}[2]{\frac{\textrm{d} #1}{\textrm{d} #2}}
\newcommand{\ud}{\textrm{d}}

\newcommand{\hp}{h^{\prime}}

\newcommand{\dpr}{c^{\prime}}

\newcommand{\cpr}{c^{\prime}}

\newcommand{\Bp}{\Theta}%

\newcommand{\Rp}{R^{\prime}}
\newcommand{\rp}{r^{\prime}}
\newcommand{\Rbar}{\bar{R}}
\newcommand{\rbar}{\bar{r}}
\newcommand{\Rmax}{R_{\textrm{max}}}
\newcommand{\rmax}{r_{\textrm{max}}}
\newcommand{\Rlbar}{\underline{R}}
\newcommand{\rlbar}{\underline{r}}
\newcommand{\Weq}{W^{(\textrm{eq})}}
\newcommand{\Wso}{W^{*}}
\newcommand{\Geq}{G^{(\textrm{eq})}}
\newcommand{\Gso}{G^{*}}
\newcommand{\xeq}{x^{(\textrm{eq})}}
\newcommand{\xso}{x^{*}}
\newcommand{\zso}{z^{*}}
\newcommand{\Rso}{R^{*}}
\newcommand{\rso}{r^{*}}
\newcommand{\X}{\mathcal X}

\newcommand{\Gresp}{G^{\textrm{(resp)}}}
\newcommand{\xresp}{x^{\textrm{(resp)}}}

\newcommand{\cti}{\tilde{c}}
\newcommand{\xl}{\underline{x}}

\newcommand{\eg}{e.g., }
\newcommand{\ie}{i.e., }

\DeclareMathOperator*{\argmax}{arg\,max}

\newtheorem{theorem}{Theorem}%
\newtheorem{proposition}{Proposition}%
\newtheorem{corollary}{Corollary}%
\newtheorem{lemma}{Lemma}%

\theoremstyle{definition}
\newtheorem{example}{Example} 

\IEEEoverridecommandlockouts

\title{\huge Incentive Mechanisms for Internet Congestion Management: Fixed-Budget Rebate \emph{versus} Time-of-Day Pricing}
\author{%
Patrick Loiseau,~\IEEEmembership{Member,~IEEE,} Galina Schwartz,~\IEEEmembership{Member,~IEEE,} John Musacchio, Saurabh Amin and  S. Shankar Sastry,~\IEEEmembership{Fellow,~IEEE,}
\thanks{Patrick Loiseau is with EURECOM, Sophia-Antipolis, France. E-mail: {\tt patrick.loiseau@eurecom.fr}. Part of this work was done while he was at UC Santa Cruz, CA, USA.}
\thanks{Galina A. Schwartz and S. Shankar Sastry are with the department of Electrical Engineering and Computer Science, UC Berkeley, CA, USA.  E-mails: {\tt \{schwartz, sastry\}@eecs.berkeley.edu}}
\thanks{John Musacchio is with the Technology and Information Management program, UC Santa Cruz, CA, USA.  E-mail: {\tt johnm@soe.ucsc.edu}}
\thanks{Saurabh Amin is with the department of Civil and Environmental Engineering, MIT, MA, USA.  E-mail: {\tt amins@mit.edu}} 
\thanks{{\color{black} The authors thank the editor and the anonymous referees for their thorough reviews which lead to significant improvements of the paper.}}
\thanks{This research was supported by NSF grant CNS-0910711 and by TRUST (Team for Research in Ubiquitous Secure Technology), which receives support from the NSF (\#CCF-0424422) and the following organizations: AFOSR (\#FA9550-06-1-0244), BT, Cisco, DoCoMo USA Labs, EADS, ESCHER, HP, IBM, iCAST, Intel, Microsoft, ORNL, Pirelli, Qualcomm, Sun, Symantec, TCS, Telecom Italia, and United Technologies.
}
}

\begin{document}

\maketitle

\begin{abstract}
Mobile data traffic has been steadily rising in the past years. This has generated a significant interest in the deployment of incentive mechanisms to reduce peak-time congestion. Typically, the design of these mechanisms requires information about user demand and sensitivity to prices. Such information is naturally imperfect. In this paper, we propose a \emph{fixed-budget rebate mechanism} that gives each user a reward proportional to his percentage contribution to the aggregate reduction in peak time demand.  For comparison, we also study a time-of-day pricing mechanism that gives each user a fixed reward per unit reduction of his peak-time demand. To evaluate the two mechanisms, we introduce a game-theoretic model that captures the \emph{public good} nature of decongestion. For each mechanism, we demonstrate that the socially optimal level of decongestion is achievable for a specific choice of the mechanism's parameter. We then investigate how imperfect information about user demand affects the mechanisms' effectiveness. From our results, the fixed-budget rebate pricing is more robust when the users' sensitivity to congestion is ``sufficiently'' convex. This feature of the fixed-budget rebate mechanism is attractive for many situations of interest and is driven by its closed-loop property, i.e., the unit reward decreases as the peak-time demand decreases.
\end{abstract}


\section{Introduction}
\label{sec.introduction}

The consumer demand for network bandwidth is steadily growing. For instance, mobile data traffic nearly tripled during each of the past three years due to increasing penetration of mobile devices such as smartphones~\cite{Cisco}. Numerous studies indicate that this growth will continue as bandwidth intensive applications like video streaming continue to gain popularity~\cite{WP}. The growing demand for bandwidth forces the Internet Service Providers (ISPs) to adopt congestion management schemes, including capacity expansion and pricing mechanisms. Although the ISPs have historically used flat-rate pricing, many ISPs are now interested in moving to tiered pricing schemes~\cite{WP, WSJ}. However, experiments have shown that users prefer flat-rate pricing, and will pay a premium to avoid being metered~\cite{Odlyzko,Varian}. This makes the adoption of real-time pricing particularly challenging. Thus, novel pricing mechanisms that balance the conflict between the need for network decongestion and the users' preference for flat prices are of great practical interest.

Network bandwidth (and hence the level of congestion) is not uniform during the course of a day; it drops at night after the prime time evening hours. This variability in demand can be exploited to design variable pricing mechanisms. For instance, time-of-day pricing mechanisms have been designed to incentivize users to shift a part of their demand to the off-peak times~\cite{Jiang08a, Wong11a}. However, such mechanisms typically require information about user demand; in particular, the knowledge of user preferences about shifting their demand from peak to off-peak times. 
In practice, this information may be inaccurate or just too difficult to obtain due to privacy concerns \cite{Wong11a}. Thus, robustness to imperfect information about user preferences must be taken into account in the design of any practically viable mechanism. 

Recently, a fixed-budget rebate mechanism (termed ``raffle-based scheme'') was proposed for decongestion of a shared resource \cite{Loiseau11a}. Decongestion is viewed as a public good: when a user reduces/shifts his demand away from peak times, his contribution benefits all the users sharing the resource. The fixed-budget rebate mechanism in \cite{Loiseau11a} is inspired by economic ideas on incentivizing contributions to provision of public goods~\cite{Morgan00a}. In this mechanism, each user is entitled a reward proportional to his percentage contribution to the total demand reduction. An attractive feature of this mechanism is that, in practice, it can be implemented via a lottery scheme, where each participating user wins a \emph{prize} with a probability equal to the fraction he contributed to the total demand reduction. 

In this article, we investigate the fixed-budget rebate mechanism, and compare it with the more traditional time-of-day pricing mechanism for reducing Internet congestion. In Sec.~\ref{sec.model}, we introduce a game-theoretic model with a continuum of non-atomic users. Each user chooses his peak time and off-peak time demand to maximize his utility. The user utility models both his benefit from peak time decongestion, and his willingness to reduce/shift away from the peak time period. The model allows us to compute the user equilibrium welfare for both mechanisms: fixed-budget rebate and time-of-day pricing. We compare their sensitivity to information imperfections for the case when an ISP with imperfect information about user demand chooses the mechanism parameters. Our results in Sec.~\ref{sec.results1} can be summarized as follows:  
\begin{itemize}
\item[(i)] For any given parameters, for each mechanism, a Nash equilibrium exists, and it is unique.
\item[(ii)] For the case when ISP has perfect information about user demand, for each mechanism, the ISP can choose the mechanism parameter to achieve the socially optimal level of decongestion. %
\item[(iii)] For the case when ISP has imperfect information about user preferences, the fixed-budget rebate mechanism is more robust to the time-of-day pricing mechanism a under mild condition on the users' sensitivity to congestion. 
\end{itemize}

Our analysis reveals several desirable features of fixed-budget rebate mechanism. First, the condition under which it is more robust than the time-of-day-pricing can be interpreted as ``convex'' user sensitivity to congestion (or delay). This condition is expected to be predominant, especially for today's Internet which supports highly delay-sensitive applications. 
This robustness of the fixed-budget rebate mechanism is driven by its closed-loop property: as the aggregate demand shifts away from peak time period, the user reward for his per unit contribution decreases. 
Finally, if an ISP decides to implement the fixed budget rebate mechanism, he knows the total reward (or rebate) that he owes to the users even when the information about user demand characteristics is imperfect. In contrast, under the time-of-day pricing mechanism, the ISP will have to design the mechanism based on an estimate of the total expected reward that he will owe to the users. 

The rest of the paper is organized as follows. Sec.~\ref{sec.related_work} discusses the related literature. We introduce the model in Sec.~\ref{sec.model}. In Sec.~\ref{sec.results1}, we analyze the two incentive mechanisms (Nash equilibrium and social optimum) and compare them in terms of robustness to imperfect information. We conclude in Sec.~\ref{sec.conclusion}. Proofs are relegated to Appendices. 

\section{Related work}
\label{sec.related_work}

Many pricing mechanisms have been proposed to manage quality of services (QoS) in networks, see e.g., surveys  \cite{Henderson01a, Tuffin03a, Sen12b}. 
For instance, in \cite{Honig95a}, Honig and Steiglitz propose a usage-based pricing mechanism, and analyze it using a model with delay-sensitive users. Their results show how to find the price that maximizes the provider's revenue by solving a fixed-point equation. A similar model is used in \cite{Basar02a} where Ba{\c s}ar and Srikant analyze the many-users limit. They show that, as the number of users increases, the provider's revenue per unit of bandwidth increases and conclude that this gives providers an incentive to increase their network capacity.  In a number of papers, e.g., \cite{Mendelson90a, Odlyzko99a, Marbach04a}, pricing mechanisms based on multiple classes of customers with different priorities are proposed and analyzed in terms of equilibrium achieved and optimal price per class. In \cite{Shen07a, Shen11a}, Shen and Ba\c{s}ar investigate the performance of non-linear pricing in a model similar to \cite{Basar02a} and find an improvement of up to 38\% over linear pricing in some cases. However, in all the aforementioned papers, the demand is assumed stationary and the price is fixed independently of the instantaneous network congestion or of the time of the day. 
In contrast, in this paper, we investigate linear pricing mechanisms that leverage the time variability of user demand using a single priority class.

A few papers have proposed mechanisms with prices dependent on congestion levels. 
In \cite{Paschalidis00a},  Paschalidis and Tsitsiklis propose a congestion-based pricing mechanism in the context of loss networks (i.e., phone). They provide a dynamic programming formulation of the revenue maximization problem and of the welfare maximization problem. Then, they show that this dynamic congestion pricing mechanism  can be well approximated by a simpler static time-of-day pricing. An alternative mechanism called ``Trade \& Cap'',   was recently proposed by Londo\~{n}o, Bestavros and Laoutaris \cite{Londono10a}. It works in two phases. First, users engage in a trading game where they choose an amount of \emph{reserved} bandwidth slots to buy for hard-constraints traffic. In the second phase, the remaining bandwidth is allocated to users as \emph{fluid} bandwidth, in proportion of their remaining ``buying power''. They show that this mechanism smoothes the aggregate demand to a certain level. In their model, users have a cost function that increases linearly with the total demand in a given slot.  
In this paper, we consider simpler one-phase  pricing mechanisms with fixed parameters.  Our model also differ from these papers in that users have elastic demand and their utility is an arbitrary function of the congestion level.

Two recent papers analyze time-of-day pricing mechanisms over $n$ time slots \cite{Jiang08a, Wong11a}. %
In \cite{Jiang08a}, Jiang, Parekh and Walrand consider a model where users have unit demand. Each user chooses one time-slot in which he transmits its entire demand, to maximize his utility. The authors of \cite{Jiang08a} obtain a bound on the price of anarchy due to users selfishness. 
In \cite{Wong11a}, Wong, Ha and Chiang consider a model with users transmitting \emph{sessions} of random length. Sessions arrive as a Poisson process and each session is characterized by a \emph{waiting function} which reflects the willingness of the user to delay his entire session for a given time, in exchange for a reward given by the provider. The authors show how to compute the optimal reward levels in order to maximize the provider profit by balancing the congestion cost due to demand exceeding capacity and the reward amount.  Further analysis of this mechanism called ``TUBE'', as well as implementation are provided in \cite{Ha12a}. However, in their model, users are only sensitive to prices (the effect of congestion on the user utility is not considered) and the analysis is not game-theoretic. 
In this paper, we consider a  model with two time slots (peak and off-peak). We provide a game-theoretic analysis.  In our model, user utility functions are the closest to \cite{Jiang08a} where user cost due to latency is an arbitrary (convex) function of the load. However, our setup differs from  \cite{Jiang08a}, as each user in our model can shift an arbitrary continuous fraction of his demand from peak time  to off-peak time.

In this paper, we show that the problem of decongesting the peak time can be seen as a public good provision problem. Our model is closely related to the ``raffle-based'' incentive mechanism, which has been recently proposed by Loiseau, Schwartz, Musacchio, Amin and Sastry \cite{Loiseau11a}. That work was  inspired by Morgan, who in \cite{Morgan00a} pioneered the investigation of using the lotteries for public good provision.
The public good perspective has been applied in recent works by Sharma and Teneketzis \cite{Sharma09a, Sharma11a} in the context of optimal power allocation for wireless networks. 
The connection of lottery-based mechanism with public good provision was originally noted in \cite{Morgan00a} and received an extensive attention in economic literature (see \cite{Lange07a}). 
The idea of lotteries has also been used in different contexts. For instance, lottery scheduling is a widely applied technique in resource scheduling in computer operating systems \cite{Waldspurger94a}. 
Recent interest in the application of lotteries to congestion management was facilitated by Merugu, Prabhakar and Rama who demonstrated with a field study that lottery-based mechanisms can be effectively used to reduce congestion in transportation systems \cite{Merugu09a}. 
In contrast, our contributions are methodological. We use a game-theoretic model to analytically study the performance of a lottery-like mechanism and compare it to a more standard time-of-day pricing mechanism. 

\section{Model}
\label{sec.model}

Let us consider a shared Internet access point with capacity $C$. Based on the usage patterns, let the day be divided into two time periods: a peak time of duration $T_p$ and an off-peak time of duration $T_o$. We assume that each time period corresponds to a stationary regime with respective loads $\rho_p$ and $\rho_o$. 

An access point is typically shared by a finite number of users, each having his own preference for time periods which we model by user \emph{type} {\color{black} (the type of a user will typically depend on the applications that he uses)}. To account for a large number of users, we model the set of users as a continuum of non-atomic users; \ie each user contributes a negligible fraction of the total demand. We use the measure-theoretic framework similar to \cite{Aumann64a, AliKhan02a}. Let $(\Theta, \mathcal{F})$ be a measurable space where $\Theta$ is the set of  user types. We assume that the user types are distributed according to a finite measure $\mu$ on $(\Theta, \mathcal{F})$\footnote{Throughout the paper, we assume that all functions of $\theta$ are measurable. In \cite{Aumann64a}, Aumann notes that ``the measurability assumption is of technical significance only and constitutes no real economic restriction.''}. 
{\color{black} While simpler modeling assumptions can be used (e.g., considering only two types), using an arbitrary measure of types $\mu$  gives a higher flexibility that can be interesting to fit real populations.}

Note that for simplicity, we describe the population at the granularity of types instead of users as in \cite{Aumann64a, AliKhan02a}. 
This is justified by the strict concavity of the user utilities (see assumptions below) which implies that at Nash equilibrium, all users of the same type choose the same action. 
{\color{black} As a consequence, although we do require that the measure of users is non-atomic (as for any non-atomic game), we do not require that the measure of types $\mu$ itself is non-atomic. For instance, if all users have the same type, measure $\mu$ is only constituted by one atom. Yet, each user of each type remains infinitesimally small, which means that the action of one user does not affect the aggregate outcome. }

\subsection{User utility}
\label{sec.utilities}

Each user of type $\theta\in\Theta$ chooses his peak-time demand $y_{\theta}$ and his off-peak time demand $z_{\theta}$ to maximize his utility
\begin{align}
\nonumber  &u_{\theta} (y_{\theta}, z_{\theta}, y_{-\theta}, z_{-\theta}) \! = \! P_{\theta} \left(y_{\theta} \right)  \! + \!  O_{\theta} \left(z_{\theta}\right) \! - \!  y_{\theta} L_p \! \left( \int_{\Theta} y_{\theta} \ud \mu (\theta) \right) \\ 
\label{eq.utility_PO} &  \quad - z_{\theta} L_o \! \left( \int_{\Theta} z_{\theta} \ud \mu (\theta) \right) - \left( y_{\theta} + z_{\theta} \right) q - p, 
\end{align}
where the notation $y_{-\theta}$ and $z_{-\theta}$ is standard: it denotes peak-time and off-peak-time demand choices for all user types but $\theta$. 
In \eqref{eq.utility_PO}, $P_{\theta} (\cdot)$ and $O_{\theta} (\cdot)$ are the utilities that a user of type $\theta\in\Theta$ gets for his demand in the peak time and off-peak time respectively. $L_p(\cdot)$ and $L_o(\cdot)$ are the disutilities due to congestion in the peak time and off-peak time respectively. These disutilities are per unit of demand, hence they are multiplied by the demand in each time. Finally, quantity $q\ge0$ is a fixed usage-based price (which could be zero) and quantity $p>0$ is a fixed monthly subscription price.

We assume that utilities $P_{\theta} (\cdot)$ and $O_{\theta} (\cdot)$ are twice differentiable increasing strictly concave functions of the demand. 
{\color{black} We assume that there is a fixed maximum peak-time demand $d_p$ (per day) that could correspond for instance to a subscription daily peak cap, that could be a maximum usable demand (determined by technology limitation), or that could be a maximum daily demand determined from empirical data}. For simplicity, we assume that this maximum peak-time demand is the same for each user, but more general cases could be handled easily\footnote{If users differ by their maximum peak-time demand, each user could be viewed as an appropriate number of users with identical maximum peak-time demand. The proposed model still applies with measure $\mu$ defined for all subset $\Theta_1\in\mathcal{F}$ by $\mu(\Theta_1) = \int_{\Delta} d\cdot \nu(\Theta_1, \ud d)$ where $\Delta$ is the set of maximum peak-time demands and measure $\nu$ on $\Theta\times\Delta$ represents the joint distribution of types and demand.} (in that case, user-dependent prices could also easily be handled). 
Each user can choose to shift to off-peak time, or to simply not use, part of his maximum peak-time demand. Additionally to the shifted peak-time demand, each user could have an initial off-peak-time demand. However, this additional demand does not modify our analysis as long as the peak time remains more congested than the off-peak time. For simplicity, we assume that the initial off-peak-time demand is zero, \ie the off-peak-time demand only corresponds to shifted peak-time demand. Then, we have the following constraint on the demands:\\[-3mm]
\begin{equation}
\label{eq.constraint} y_{\theta} + z_{\theta} \le d_p, \quad  (\theta\in\Theta). \\[-1mm]
\end{equation}

We assume that disutilities $L_{p} (\cdot)$ and $L_{o} (\cdot)$ are increasing strictly convex functions of the aggregate demand in each time (a similar assumption is made, \eg  by Jiang, Parekh and Walrand \cite{Jiang08a}). This assumption is realistic and quite general. As an example, let us focus on the average delay $\delta$ as a measure of the network quality, as in Honig and Steiglitz \cite{Honig95a}. Our assumption holds if ($i$) the disutility is an increasing convex function of the average delay and ($ii$) the average delay is an increasing strictly convex function of the aggregate demand %
or equivalently of the load in the corresponding time: \vspace{-2mm}
\begin{equation*}
\rho_p = (CT_p)^{-1}  \int_{\Theta} y_{\theta} \ud \mu (\theta),  \;  \; \rho_o = (CT_o)^{-1}  \int_{\Theta} z_{\theta} \ud \mu (\theta).
\end{equation*}
Assumption ($i$) is natural: an increase of the delay from zero to half a second creates no more disutility than from half a second to one second. This assumption is also made in \cite{Honig95a}. 
Assumption ($ii$) holds for the vast majority of queueing models considered in the literature. For example, it holds for the processor sharing queue (the most classical model for 3G and 4G networks \cite{Borst03a}), for which the average delay is \\[-3mm]
\begin{equation}
\label{eq.delay}
\delta\left( \rho_p \right) = \frac{\delta_0}{1-\rho_p}, \\[-2mm]
\end{equation}
where $\delta_0$ is a constant. It also holds for common models of wired networks such as the M/D/1 model considered in \cite{Honig95a} and the M/M/1 model considered by Shen and Ba\c{s}ar \cite{Shen07a, Shen11a}. 
Finally, we assume that despite the effect of users shifting part of their demand, the off-peak time remains relatively uncongested so that \\[-3mm]
\begin{equation}
\label{eq.L0}
L_o \left(  \int_{\Theta} z_{\theta} \ud \mu (\theta) \right) \simeq 0. \\[-1mm]
\end{equation}
This assumption is not strictly necessary but greatly simplifies the presentation without affecting the important effects that we consider in this model\footnote{If one wants to consider a  non-zero off-peak-time disutility, this assumption could be replaced by the relaxed assumption that when the aggregate shifted demand increases, the marginal peak-time disutility reduction is higher than the marginal off-peak-time disutility increase.}.

For numerical illustrations of our model, we use the following example of an Internet access point.
\begin{example}[Internet access point]
\label{ex.model_use_case}
The capacity is $C=1$~Gbps. Peak-time lasts $T_p=2$~h (\eg 6pm to 8pm), hence $T_o = 22$~h.  $\Theta = [0, 1]$ with a uniform distribution of types $\mu(\ud \theta) = D_p/d_p\cdot \ud \theta$, where $D_p=7.2\cdot 10^3$~Gbits and $d_p=7.2$~Gbits (which corresponds to $1,000$ users with peak-time capacity $1$~Mbps). 
The latency disutility is $L_p(\rho_p) = L_0 \delta (\rho_p)$ where $\delta (\rho_p)$ is given by the PS model~\eqref{eq.delay} and $L_0=0.065$. Peak-time utility is $P_{\theta}(x_{\theta}) = (1+\theta) P_0 \log(1+x/d_p)$ with $P_0=130$ and off-peak time utility is $O_{\theta} = 1/10 \cdot P_{\theta}(\cdot)$. The subscription price is $p=\$50$ and the usage-base price is zero.
\end{example}

\subsection{User type distribution}
\label{sec.measure}

On the timescale of a day, the population is heterogeneous with user types distributed according to measure $\mu$. However, we assume that each user has a type that varies randomly across the days of a month, with the same distribution $\mu$. Therefore, the population is homogeneous in average on the timescale of a month. In particular, with this assumption, the expected utility of each user on the timescale of a month equals the daily aggregate welfare \\[-3mm]
\begin{equation}
\label{eq.defW} W = \int_{\Theta} u_{\theta} \ud \mu ( \theta) \\[-4mm]
\end{equation}
normalized by $\mu(\Theta)$. 

Each user will buy a monthly contract (with subscription price $p$) to use the service if his expected utility over the month is positive, \ie here if\\[-4mm]
\begin{equation}
\label{eq.participation}
W>0.\\[-3mm]
\end{equation}

We assume that without any incentive mechanism, this condition is satisfied. Then, our assumption guarantees that with any welfare-improving incentive, each user will continue to participate, \ie to buy the monthly contract. 

{\color{black} Note that if the population cannot be assumed homogeneous at the timescale of a month, it is possible to divide it into subpopulations that can be assumed homogeneous and to apply our incentive mechanisms to each of these subpopulations. }

\subsection{Model reduction to one-dimensional strategy space}
\label{sec.reduction}

Before introducing the incentive mechanisms, we show that our model can be reduced to a one-dimensional strategy space focusing on the peak-time demand reduction.
With assumption~\eqref{eq.L0}, the utility~\eqref{eq.utility_PO} of a user of type $\theta\in\Theta$ can be re-written as \vspace{-2mm}
\begin{align}
\nonumber  u_{\theta} (y_{\theta}, z_{\theta}, y_{-\theta}) =  & P_{\theta} \left(y_{\theta} \right)  -  y_{\theta} \cdot L_p \left( \int_{\Theta} y_{\theta} \ud \mu (\theta) \right)  - y_{\theta} \cdot q\\
\label{eq.utility_PO2} & + O_{\theta} \left(z_{\theta}\right) - z_{\theta} \cdot q -p. 
\end{align}
Since we are interested in the reduction of peak time demand, we define the difference between the maximum peak-time demand and the chosen peak-time demand: \vspace{-2mm}
\begin{equation}
x_{\theta} = d_p - y_{\theta}, \quad (\theta\in\Theta).\\[-1mm]
\end{equation}
This peak-time demand \emph{reduction} includes both the unused peak-time demand and the peak-time demand shifted to off-peak time. 
For a given $x_{\theta}\in [0, d_p]$, we define the optimal shifted demand:\vspace{-2mm}
\begin{equation*}
z^{*}_{\theta} (x_{\theta}) = \argmax_{z_{\theta}\in[0, x_{\theta}]} \left[ O_{\theta}(z_{\theta}) - z_{\theta}q \right], \quad (\theta\in\Theta). \\[-1mm]
\end{equation*}
A user of type $\theta\in\Theta$ maximizing his utility~\eqref{eq.utility_PO2} will choose a couple $(x_{\theta}, z_{\theta})$ such that $z_{\theta} = \zso_{\theta}(x_{\theta})$.\footnote{A function $(x_{\cdot}, z_{\cdot})$ corresponding to social welfare maximization also satisfies $z_{\theta} = \zso_{\theta}(x_{\theta})$ for all $\theta\in\Theta$.} 
As we are interested in the reduction of congestion at peak-time, we restrict our attention to the choice of $x_{\theta}$. Note that if $q=0$, then $\zso(x_{\theta}) = x_{\theta}$. Indeed, if there is no usage-based cost, off-peak-time demand always gives higher utility than $0$.

In the absence of latency, the maximal utility of a user is \vspace{-1mm}
\begin{equation}
\label{eq.Ubar}
\bar{u}_{\theta} \!=\!    P_{\theta} \left(d_p - \xl_{\theta} \right)  \!-\!  \left(d_p-\xl_{\theta}\right) q \!+\! O_{\theta} \left(z^{*}_{\theta}(\xl_{\theta})\right) \!-\! z^{*}_{\theta}(\xl_{\theta}) q,\\[-1mm] %
\end{equation}
for all $\theta\in\Theta$,  where \vspace{-2mm}
\begin{align*}
\xl_{\theta} = \argmax_{x_{\theta}\in [0, d_p]}  \Big\{  & P_{\theta} \left(d_p - x_{\theta} \right)  -  \left(d_p-x_{\theta}\right) q \\[-2mm]
& + O_{\theta} \left(z^{*}_{\theta}(x_{\theta})\right) - z^{*}_{\theta}(x_{\theta}) q \Big\}, \quad (\theta\in\Theta), 
\end{align*}
is the baseline peak-time demand reduction which maximizes the latency-free utility. 
Latency and incentive mechanisms will only result in users using less of their peak-time demand, \ie increasing their choice of $x_{\theta}$ beyond $\xl_{\theta}$. 
Then, we define the \emph{cost of shifting} as the loss of utility incurred by a user when reducing his peak-time demand:\vspace{-2mm}
\begin{align}
\label{eq.cost}
c_{\theta} (x_{\theta})  =    \bar{u}_{\theta} - & \Big[P_{\theta} \left(d_p - x_{\theta} \right)  -  \left(d_p-x_{\theta}\right) q \\
\nonumber & \;\;\;\; + O_{\theta} \left(z^{*}_{\theta}(x_{\theta})\right) - z^{*}_{\theta}(x_{\theta}) q \;\Big], \quad (\theta\in\Theta).
\end{align}
(Note that with a slight abuse of terminology, we call $c_{\theta} (x_{\theta})$ the cost of \emph{shifting} whereas the peak-time demand reduction $x_{\theta}$ can actually correspond to shifted demand and/or to unused demand.) 
The definition of the baseline~\eqref{eq.Ubar} guarantees that $c_{\theta} (x_{\theta})$ is always positive. 
Moreover, with our assumptions on functions $P_{\theta}(\cdot)$ and $O_{\theta}(\cdot)$, the cost of shifting $c_{\theta}(\cdot)$ is differentiable and strictly convex on $[0, d_p]$; and increasing on $[\xl_{\theta}, d_p]$ (see details in Appendix~\ref{app.costconvex}). Finally, to simplify the proofs, we assume that the marginal cost of shifting is bounded by a constant independent of $\theta$. %

We view the aggregate peak-time demand reduction \vspace{-2mm}
\begin{equation}
\label{eq.defG}
G = \int_{\Theta} x_{\theta}   \ud \mu ( \theta) \\[-1mm]
\end{equation}
as a \emph{public good} to which each user contributes by his choice of $x_{\theta}$.  
Indeed, when a user reduces his peak-time demand, the benefits of reduced peak-time congestion is shared by all the users. 
We define the function \vspace{-1mm}
\begin{equation}
\label{eq.defh}
h(G) = - L_p\left(D_p - G\right),\\[-1mm]
\end{equation}
where $D_p = d_p \mu(\Theta)$
is the aggregate maximum peak-time demand. 
Function $h(\cdot)$ reflects the notion of how much users benefit from the network decongestion at peak-time. 
With our assumptions on $L_p(\cdot)$, $h(\cdot)$ is an increasing concave function of the public good level $G$. 
The term  $- y_{\theta} L_p\left(\int_{\Theta} y_{\theta} \ud \mu (\theta)\right) = (d_p-x_\theta) h(G)$ in \eqref{eq.utility_PO2} has the interpretation that the benefit a user gets from the peak-time decongestion is the product of his peak-time demand $(d_p-x_\theta)$ times the benefit per unit demand $h(G)$.
Notice that $h(G)$ is negative, but its most important characteristic is that it is increasing in $G$, \ie the disutility due to congestion reduces when $G$ increases.

In summary, in view of \eqref{eq.utility_PO2}-\eqref{eq.defh}, our peak-time decongestion model reduces to a public good provision problem similar to \cite{Loiseau11a}: the utility of a user of type $\theta\in\Theta$ is \vspace{-1mm}
\begin{equation}
\label{eq.utility}
u_{\theta}(x_{\theta}, G) = \bar{u}_{\theta} +(d_p-x_{\theta}) h (G) - c_{\theta} (x_{\theta}) - p, \\[-1mm]%
\end{equation}
where $h(\cdot)$ corresponds to the (unit) benefit from the public good and $c_{\theta}(\cdot)$ corresponds to the cost of contribution. 
From our assumptions, these functions satisfy: \\[-5.5mm]
\begin{trivlist}
\item[\bf (A1)] $h(\cdot)$ is twice differentiable, strictly concave and increasing on $[0, D_p]$;
\item[\bf (A2)] $c_{\theta}(\cdot)$ is positive, differentiable and strictly convex on $[0, d_p];$ and increasing on $[\xl_{\theta}, d_p],$ $\quad$ ($\theta\in\Theta$);
\item[\bf (A3)] $\sup_{\theta\in\Theta} \cpr_{\theta} (d_p) < \infty$, \quad $(\theta\in\Theta)$.
\end{trivlist}

\vspace{-1mm}

{\color{black} Notice that in our model,  $h(\cdot)$  does not depend on the type. All the type-dependency is carried by the cost of shifting. This modeling choice ensures tractability of the equilibrium.  }

\subsection{Incentive mechanisms}
\label{sec.mechanisms}

Individual users maximize their own utility \eqref{eq.utility}, which differs from maximizing \eqref{eq.defW}. Thus, in general, the level of public good and the aggregate user welfare achieved in the individual maximization and in the social optimum differ.

To align Nash equilibrium and social optimum objectives, the service provider can design mechanisms to incentivize users to reduce their peak-time demand. 
In this paper, we compare two different incentive mechanisms: a fixed-budget rebate mechanism (denoted $R$ or FBR) and a time-of-day pricing mechanism (denoted $T$ or TDP). 
Each mechanism introduces a reward based on the peak-time demand reduction $x_{\theta}$ below the maximum $d_p$. For the service provider to finance the respective reward, each mechanism also introduces an increase in the subscription price. However, as we will see (Corollary~\ref{cor.1}), each user's net utility can be improved even with this price increase.  
With mechanism $j\in\{R, T\}$, the user utility becomes \vspace{-2mm}%
\begin{equation}
\label{eq.utility_general}
u_{\theta}^j (x_{\theta}, G) = u_{\theta} (x_{\theta}, G) + M^{j}(x_{\theta}, G), \quad (\theta\in\Theta). \\[-1mm]
\end{equation}

The fixed-budget rebate mechanism consists in giving each user a reward proportional to his fraction of the total contribution, \ie of the functional form: \vspace{-1mm}
\begin{equation}
\label{eq.ESraffle}
M^R(x_{\theta}, G) = R \cdot \frac{x_{\theta}}{G}- \Delta p_R, \;\;  \textrm{ [fixed-budget rebate]} \\[-1mm]
\end{equation}
where $R$ is a parameter of the mechanism chosen \emph{ex-anti} by the provider. 
In practice, this mechanism could be implemented via randomization. 
For example, with a finite number of users, it could be implemented by the simplest type of lottery where each user wins the prize $R$ with a probability equal to his percentage contribution to the total amount of peak-time demand reduction.  
In this case, \eqref{eq.utility_general} and \eqref{eq.ESraffle} would correspond to expected utilities. 
Other implementations (\eg deterministic) are also  possible. 
To complete the definition of the fixed-budget rebate mechanism, we assume that if no user reduces his peak-time demand then the reward is not given. 
However, if the set of users who reduce their peak-time demand is nonempty but of measure zero, then each contributing user receives an infinite reward given in such a way that the integral w.r.t. the measure of users is $R$.  
{\color{black} This is a technical assumption for the measure-theoretic setting of the non-atomic game. In practice, it reflects the fact that if only a finite number of users contribute, their expected reward relative to their fraction of the total demand grows to infinity as the total number of users goes to infinity. }

{\color{black} We notice here that the fixed-budget rebate mechanism introduces uncertainty in the users bill as the reward depends on the amount shifted by the other users. However, this uncertainty is only one-sided: the maximum bill is known and only the reward amount is uncertain. This asymmetry is crucial to ensure good adoption of the mechanism.  }

The time-of-day pricing mechanism corresponds to a fixed reward per unit of shifted demand: \vspace{-1mm}
\begin{equation}
\label{eq.ESdifferentiated}
M^T(x_{\theta}, G) = r \cdot x_{\theta}- \Delta p_T, \;\; \textrm{ [time-of-day pricing]}\\[-1.5mm]
\end{equation}
where $r$ is a parameter of the mechanism chosen \emph{ex-anti} by the provider. 
This mechanism is a variation of a conventional time-of-day pricing mechanism, with an off-peak price subsidy. Its implementation is straightforward.

In \eqref{eq.ESraffle} and \eqref{eq.ESdifferentiated}, $\Delta p_j$ denotes the increase in the subscription price that the service provider imposes to finance the reward mechanism.
Let $\Geq$ be the equilibrium level of public good (in the next section, we show that the Nash equilibrium is unique for both mechanism). 
We assume that the price $\Delta p_j$ is fixed in advance by the service provider to compensate the reward, \ie such that $\int_{\Theta} M^j(x_{\theta}, \Geq) \ud \mu (\theta) = 0$ (note that the expression of  the aggregate welfare~\eqref{eq.defW} is thus not directly modified by the mechanisms, but only through the chosen contributions $x_{\theta}$)\footnote{If users have different maximum peak-time demand for which they are charged different subscription prices, it is also possible to impose a type-dependent price increase $\Delta p_{\theta}$ which compensate the reward, \ie such that $\int_{\Theta} M^j(x_{\theta}, \Geq) \ud \mu (\theta) = 0$ is still satisfied. }. Then, \vspace{-2mm}
\begin{equation}
\label{eq.defp}
\Delta p_R = R \cdot \frac{d_p}{D_p} \quad \textrm{ and } \quad \Delta p_T = rG^{(eq)} \cdot \frac{d_p}{D_p}. \\[-2mm]
\end{equation}
From \eqref{eq.defp}, we immediately see that the service provider has to know the equilibrium to determine the price $\Delta p_T$ for the time-of-day pricing mechanism. An error in the estimation of $\Geq$ could have important consequences. In contrast, such knowledge is not necessary for the fixed-budget rebate mechanism where $\Delta p_R$ only depends on the parameter $R$ chosen by the service provider. 

The {\color{black} marginal } utility with mechanism $j\in\{R, T\}$ is \vspace{-2mm}%
\begin{equation}
\label{eq.marginal}
\cd{u_{\theta}^j}{x_{\theta}} = - h (G) - \dpr_{\theta} (x_{\theta}) + {M^j}^{\prime}(G), \quad (\theta\in\Theta),\\[-2mm]
\end{equation}
where%
\begin{equation}
\label{eq.Sp}
{M^R}^{\prime}(G) = \frac{R}{G} \quad \textrm{ and } \quad {M^T}^{\prime}(G) = r. \\[-2mm]
\end{equation}
{\color{black} Notice that in \eqref{eq.marginal}, there is no $h^{\prime}(G)$ term corresponding to the variation of the aggregate due to the variation of a user's decision. This is because, in the non-atomic game, users are negligible and do not account for the variation of $G$ induced by their action when taking their decision. In \eqref{eq.marginal} and in all future occasions, we abuse notation by denoting with a partial derivative w.r.t. $x_{\theta}$ the marginal quantities corresponding to variations following the variation of a user's action.   }

For both mechanisms, the marginal reward ${M^j}^{\prime}$ is independent of the individual contribution $x_{\theta}$. 
Due to the term $-h(G)$ in \eqref{eq.marginal}, the marginal utility decreases when $G$ increases. Intuitively, if the congestion is lower at peak time, a user would want to use it more. Hence he would want to reduce less his peak-time demand. 
This decrease of the marginal utility is accentuated by the term ${M^R}^{\prime}(G) = R/G$ in the case of the fixed-budget reward mechanism.

\section{Analysis}
\label{sec.results1}

In this section,\footnote{Some of the first results of this section appeared in \cite{Loiseau11a} for the fixed-budget rebate mechanism. They are extended here to handle both mechanisms. } we show that, for each mechanism, there exists a unique Nash equilibrium. Then, we show that for appropriate values of the mechanisms parameters, they achieve social optimum and that for a wide range of parameters, both mechanisms are welfare improving. Finally, we compare the two mechanisms based on their sensitivity to imperfect information about the user utilities.

For clarity, we will use the following notation: $\Gamma_{R} \left( \Theta, \mu, h, \{ c_{\theta}\}_{\theta\in\Theta}, R \right)$ and $\Gamma_{T} \left( \Theta, \mu, h, \{ c_{\theta}\}_{\theta\in\Theta}, r \right)$ are the non-atomic games where users selfishly optimize their own utility \eqref{eq.utility_general} in the fixed-budget rebate mechanism and in the time-of-day pricing mechanism respectively. 
We denote with the superscript $^{(\textrm{eq})}$ the quantities at equilibrium in both games and we explicitly write their dependence on $r$ and $R$ or on other parameters whenever necessary to avoid ambiguity. Similarly, we denote with the superscript $^{*}$ the social optimum quantities  corresponding to the maximization of~\eqref{eq.defW}, and denote explicitly their dependence on parameters whenever necessary.

\subsection{Nash equilibrium existence and uniqueness}
\label{sec.NE}

We define a Nash equilibrium of the non-atomic game $\Gamma_{j}$ ($j\in\{ R, T\}$) as a function $\xeq: \Theta\to [0, d_p]$ such that for all $\theta\in\Theta$, $u_{\theta}^j(x_{\theta}, \xeq_{-\theta}) \leq u_{\theta}^j(\xeq_{\theta}, \xeq_{-\theta}), \forall x_{\theta} \in [0, d_p]$. %
Due to the strict concavity of the utility $u_{\theta}^j$, it is equivalent to $\xeq$  satisfying the first-order conditions (FOCs)\vspace{-2mm}
\begin{equation}
\label{eq.FOC_continuous}
\cd{u_{\theta}^j}{x_{\theta}} \;\;
\left\{ 
\begin{array}{l}
\le 0, \quad \forall \theta: x_{\theta} = 0, \\
= 0, \quad \forall \theta: x_{\theta} \in (0, d_p), \\
\ge 0, \quad \forall \theta: x_{\theta} = d_p, %
\end{array}
\right.\\[-3mm]
\end{equation}
where $\cd{u_{\theta}^j}{x_{\theta}}$ is given by \eqref{eq.marginal}, and satisfying \eqref{eq.defG}.

The first  theorem establishes existence and uniqueness of the Nash equilibrium for both incentive mechanisms. %
\begin{theorem}
\label{thm.Nash}
For the fixed-budget rebate mechanism, for any $R\ge 0$, there exists a unique Nash equilibrium $\xeq(R)$. %

The same result holds for the time-of-day pricing mechanism, for any $r\ge 0$. 
\end{theorem}

Intuitively, for a given level of public good $G$, each user of type $\theta\in\Theta$ chooses his best response contribution $\xresp_{\theta}(G)\in[0,d_p]$ to maximize his utility.
Then, integrating the contribution of each type gives the amount of public good $\Gresp(G)$ that users want to provide in response to a given $G$. An equilibrium occurs when both quantities are equal, which corresponds to solving the fixed-point equation\vspace{-1mm}
\begin{equation}
\label{eq.fixed-point}
\Gresp (G) = G.
\end{equation}

\vspace{-1mm}
Fig.~\ref{fig.case1} illustrates the two terms of the fixed-point equation for both mechanisms. 
As we mentioned, a key feature of our model is that the higher the level of public good $G$ is (\ie the lower peak-time congestion is), the fewer users are willing to reduce their peak-time demand (the marginal utility~\eqref{eq.marginal} is decreasing in $G$). 
Therefore the aggregate best response  $\Gresp(G)$ decrease when $G$ increases and this decrease is faster for the fixed-budget rebate mechanism for which the marginal utility decreases faster. 
Moreover, $\Gresp(G)$ is continuous, which leads to a unique fixed point. The continuity of $\Gresp(G)$ is due to assumption \textbf{(A2)} (a linear cost of shifting could induce discontinuities where a slight modification of $G$ would make some users switch from not reducing their peak-time demand to reducing it by $d_p$).

\begin{figure}
\centering
\begin{tabular}{cc}
\rotatebox{90}{\hspace{1.9cm} $G, \Gresp$} & \includegraphics[width=0.75\columnwidth]{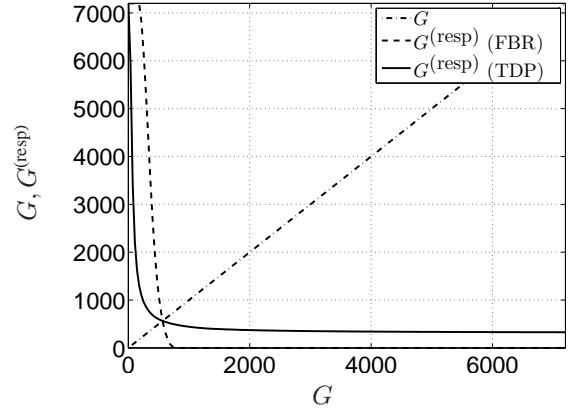}  \\ 
& $G$
\end{tabular}
\caption{Illustration of the fixed-point equation \eqref{eq.fixed-point} for Example~\ref{ex.model_use_case}. 
The dashdotted line corresponds to the amount $G$ that users shifts  (r.h.s. of \eqref{eq.fixed-point}). The dashed and solid lines correspond to the aggregate amount $\Gresp$ that users would want to shift (l.h.s. of \eqref{eq.fixed-point}) given $G$ has been shifted for the fixed-budget rebate and time-of-day pricing mechanisms respectively. To obtain  $\Geq = \Gso=78,000$, parameters were set to $R=\$5,500$ and $r=\$9$/Gbit. The corresponding subscription price increase is $\$5.5$.
\vspace{-5mm}
}
\label{fig.case1}
\end{figure}

\subsection{Social optimum}
\label{sec.SO}

We now show that the social optimum is unique and coincides with the Nash equilibrium of both mechanisms for parameters $\Rso$ and $\rso$ given in the next theorem. 
\begin{theorem}
\label{thm.SO}
The following characterizes the social optimum:
\begin{enumerate}[(i)]
\item%
There exists a function $\xso$, uniquely determined almost-everywhere,  which maximizes the aggregate welfare \eqref{eq.defW}.

\item%
For the fixed-budget rebate mechanism, we have $\xeq(R) = \xso$ almost-everywhere (and hence $\Geq(R) = \Gso$) for $R=\Rso$, where\vspace{-1mm}
\begin{subequations}
\label{eq.rRso}
\begin{equation}
\label{eq.Rso}\Rso =  \Gso \hp (\Gso) (D-\Gso).  
\end{equation}

The same result holds for the time-of-day pricing mechanism for  $r=\rso$, where
\begin{equation}
\label{eq.rso}\rso =  \hp (\Gso) (D-\Gso).  
\end{equation}
\end{subequations} 
\end{enumerate}
\end{theorem}

Intuitively, this result holds because the externality faced by a user ($-h(G)+{M^j}^{\prime}$) in the game corresponding to any mechanism is independent of his type. Therefore, by fixing a reward that is also independent of the type, it is possible to achieve social optimum (similarly to a Pigovian tax \cite{Pigou32a}). %

For Example~\ref{ex.model_use_case}, Tab.~\ref{tab.1} illustrates the effect of the incentive mechanisms with the optimal parameters of Theorem~\ref{thm.SO}: they permit a $180\%$ increase of the aggregate welfare which, in our model, also correspond to a $180\%$ increase of the average utility of each user over the timescale of a month. Peak-time congestion is significantly decreased: the load is decreased by $7\%$ but the average delay drops by $90\%$. On the other hand, off-peak time decongestion is hardly increased. 

\begin{table}
\centering
\caption{Effect of the incentive mechanisms on congestion for Example~\ref{ex.model_use_case} (cf. Fig.~\ref{fig.case1}). The right column correspond to any of the two mechanism with its optimal parameters $\Rso=\$5,500$ and $\rso=\$9$/Gbit (\ie to social optimum).\vspace{-3mm} }
\label{tab.1}
\begin{tabular}{|c|c|c|}
\hline
& \multirow{3}{*}{no incentive mechanism} & incentive mechanism \\
& & with optimal parameter\\
& & (= social optimum)\\
\hline
\hline
G & 55~Gbits & 565~Gbits \\
W & 28,000 & 78,000 \\
\hline
\hline
$\rho_p$ & 0.99  & 0.92 \\
$\delta(\rho_p)$ & 130~s & 12~s \\
\hline
\hline
$\rho_o$ & 0.092 & 0.098 \\
$\delta(\rho_o)$ & 1.10~s & 1.11~s \\
\hline
\end{tabular}
\vspace{-3mm}
\end{table}

\subsection{Nash equilibrium variation with the mechanism parameters}
\label{sec.NEvar}

In this section, we investigate the variation of the equilibrium quantities when the mechanism parameters $r$ and $R$ vary. 
For ease of exposition, we first assume that the participation constraint~\eqref{eq.participation} is not imposed (we will come back to the effect of the participation constraint later in this section, see Proposition~\ref{prop.constraint}). Then, we have the following results on the variations of the equilibrium contributions. 
\begin{proposition}
\label{prop.variation_rR}
If the participation constraint~\eqref{eq.participation} is not imposed, for the fixed-budget rebate mechanism, we have: 
\begin{enumerate}[(i)]
\item  For any $\Rp>R$, $\xeq_{\theta}(\Rp)\ge \xeq_{\theta}(R)$ ($\forall \theta\in\Theta$); and the inequality is strict if  $0<\xeq_{\theta}(R)<d_p$.

\item  For any $\Rp>R$, $\Geq(\Rp)\ge \Geq(R)$; and the inequality is strict if  $0<\Geq(R)<D_p$.

\item There exists a threshold $\Rbar>\Rso$ such that,  for any $R\ge \Rbar$, $\xeq_{\theta}(R) = d_p$ for all $\theta\in\Theta$ and $\Geq(R)=D_p$.\footnote{To avoid ambiguity on the definition of the thresholds $\rbar, \Rbar$, we assume that they are the smallest possible such thresholds.}

\end{enumerate}

The same results hold for the time-of-day pricing mechanism by changing $R$ to $r$ everywhere. 
\end{proposition}

Intuitively, since the marginal utility \eqref{eq.marginal} increases with the reward parameters, the equilibrium contributions of each user increases (result \textit{(i)}); and similarly for the equilibrium level of public good (result \textit{(ii)}).
The existence of the  thresholds $\Rbar$ and $\rbar$ (result \textit{(iii)}) is a consequence of assumption \textbf{(A3)} which means that reducing even the last  bit  of his peak-time demand implies a finite marginal cost for the user, which can be compensated by a large-enough reward. Clearly, a case with such a large reward will not happen in practice, nevertheless we include it here for completeness of the model analysis.

Proposition~\ref{prop.variation_rR} implies that for large enough parameters, the equilibrium level of public good will be positive. Let us define, for the fixed-budget rebate mechanism, $\Rlbar$ as the smallest parameter value such that $\Geq(R)>0$ for $R>\Rlbar$; and similarly $\rlbar$ for the time-of-day pricing mechanism. Then we have the following result characterizing these thresholds.
\begin{proposition}
\label{prop.raffle}
For  the fixed-budget rebate mechanism,  $\Rlbar=0$, \ie $\Geq(R)>0$ for any $R>0$ (if the participation constraint~\eqref{eq.participation} is not imposed). 

For the time-of-day pricing mechanism, $\rlbar\ge0$. 
\end{proposition}

The intuition behind Proposition~\ref{prop.raffle} is as follows. For the fixed-budget rebate mechanism, for any $R>0$, the marginal reward is infinite at $G=0$. All users want to contribute hence this is not an equilibrium. 
In contrast, for the time-of-day pricing mechanism, the marginal reward is constant. If it is  small enough so that the marginal utility of almost-all user types is non-positive at  $G=0$, then it is the equilibrium.

Note that Proposition~\ref{prop.raffle} holds independently of the value of $\Gso$ and is consistent with Theorem~\ref{thm.SO}. In particular, if $\Gso = 0$, then social optimum is achieved at Nash equilibrium for the fixed-budget rebate  mechanism only for $R=\Rso=0$; whereas social optimum is achieved at Nash equilibrium for the time-of-day pricing mechanism for any $r$ smaller than $\rlbar$.

The next proposition describes the evolution of the aggregate welfare with the mechanism parameters. 
\begin{proposition}
\label{prop.W}
If the participation constraint~\eqref{eq.participation} is not imposed, for the fixed-budget rebate mechanism,  the equilibrium aggregate welfare $\Weq(R)$ is increasing in $[0, \Rso]$, decreasing in $[\Rso, \bar{R}]$ and constant for $R\ge \bar{R}$. 

For the time-of-day pricing mechanism, the equilibrium aggregate welfare $\Weq(r)$ is constant on $[0, \rlbar]$. For $r\ge\rlbar$, the same results as for the fixed-budget rebate mechanism hold by changing $R$ to $r$ everywhere. 
\end{proposition}

Proposition~\ref{prop.W}, illustrated on Fig.~\ref{fig.WrR} shows that the welfare is unimodal. If $\Gso>0$, it increases to its only maximum at $\Rso$ or $\rso$ and then decreases. If $\Gso=0$ (hence $\Rso=0$ and $\rso=\rlbar$), the welfare is maximal at $R=0$ or $r=0$ (\ie with no incentive mechanism) and it only decreases (after a constant phase for the time-of-day pricing mechanism).

In extreme cases where the reward parameter is too large, the equilibrium aggregate welfare may become negative. 
For instance, consider a case where the usage-based price $q$ is so high compared to the off-peak time utility $O_{\theta}(\cdot)$ that all users have zero off-peak time demand. If the reward is larger than $\Rbar$ or $\rbar$, then users would not use the service at all and the aggregate welfare would be $-p\mu(\Theta)<0$. In that case, the participation constraint is not satisfied, hence users will simply not buy the service. The next proposition, which is easily derived using the monotonicity of Proposition~\ref{prop.W}, describes how the previous results are changed when introducing the participation constraint. 
\begin{proposition}
\label{prop.constraint}
If the participation constraint~\eqref{eq.participation} is imposed, for the fixed-budget rebate mechanism, there exists a threshold $\Rmax\in(\Rso, \infty]$ such that 
\begin{enumerate}[(i)]
\item For all $R<\Rmax$, all the users buys the monthly subscription and the results of Proposition~\ref{prop.variation_rR}, Proposition~\ref{prop.raffle} and Proposition~\ref{prop.W} hold. 

\item For all $R>\Rmax$, no user buys the monthly subscription, hence the welfare is zero.
\end{enumerate}

The same results hold for the time-of-day pricing mechanism by changing $R$ to $r$ everywhere. 
\end{proposition}

The effect of the participation constraint is simple: below a threshold $\Rmax$ or $\rmax$, all the users participate and above this threshold, no users participate. This is due to our assumption that the population is homogeneous at the timescale of a month. Since users are offered a monthly subscription, they will buy it if they expected utility over the month is positive which is equivalent to the aggregate welfare being positive. Due to our assumption that the welfare is positive without any incentive mechanism, we have $\Rmax>\Rso$, \ie the welfare is positive for any $R\le\Rso$. The threshold $\Rmax$ can even be infinite if the off-peak time utility is high enough and the usage-based price is small enough so that users have positive utility over the month even without using the peak time.

The last result, which is a direct consequence of the previous results of this section, shows that both mechanisms are welfare improving for a wide range of parameters. 
\begin{corollary}
\label{cor.1}
If $\Gso>0$, the fixed-budget rebate mechanism is strictly welfare improving for any parameter $R$ in a range $(0, R_0)$ where $R_0\in(\Rso, \Rmax]$: \vspace{-1mm}
\begin{equation*}
\Weq(R) > \Weq(0), \quad \forall R \in (0, R_0).\\[-1mm]
\end{equation*}

The same results hold for the time-of-day pricing mechanism by changing $R$ to $r$ everywhere, and $0$ to $\rbar$. 
\end{corollary}

This result is important as it shows that, by implementing an incentive mechanism with a parameter lying in a wide range around an optimal parameter, the provider can increase welfare. Fig.~\ref{fig.WrR} shows that for Example~\ref{ex.model_use_case}, the time-of-day pricing mechanism with any parameter in $(0, 2\rso)$ is welfare improving, and the fixed-budget rebate mechanism with any parameter in $(0, 14\Rso)$ is welfare improving
However, a consequence of Proposition~\ref{prop.W} is that both mechanisms can ``overshoot'': if $R$ or $r$ is too large (larger that $\Rso$ or $\rso$), $\Geq$ can be larger than $\Gso$ and the aggregate user welfare is suboptimal.
In a competitive environment, a provider would not intentionally choose an overshooting parameter because it would be a competitive disadvantage as compared to a provider choosing an optimal parameter. 
However, if the provider has imperfect information about user utilities, it may overshoot unintentionally. Fig.~\ref{fig.WrR} suggests that in this case, the aggregate welfare remains higher for the fixed-budget rebate mechanism than for the time-of-day pricing mechanism. In the next section, we investigate in details the robustness of each mechanism to imperfect information about user utilities.

\begin{figure}
\begin{tabular}{ccc}
\rotatebox{90}{\hspace{5.4cm} \rotatebox{-90}{(a)}} & \rotatebox{90}{\hspace{1.8cm} $\Weq(R)$} \hspace{-0.4cm} & \includegraphics[width=0.75\columnwidth]{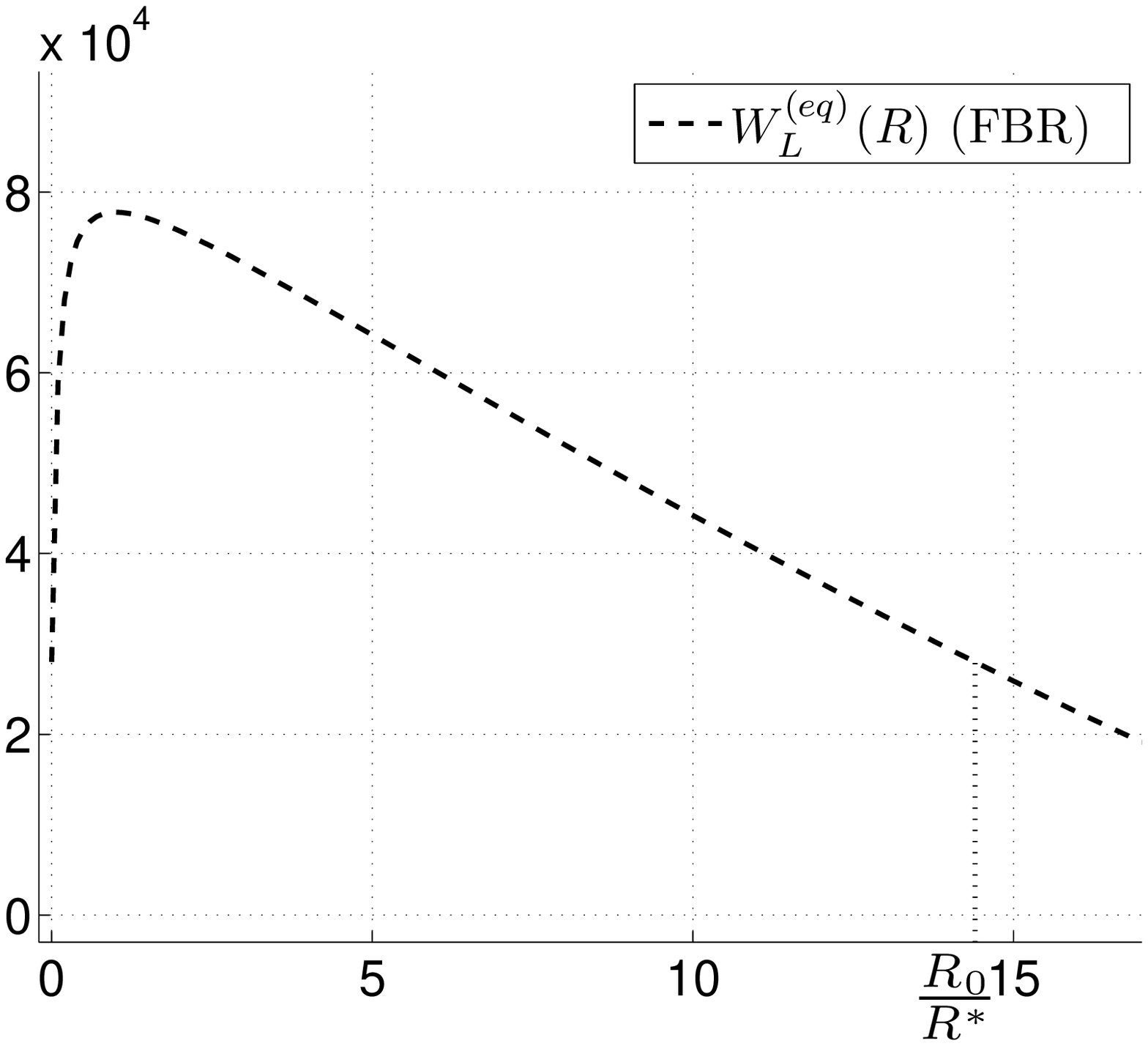} \\[-3mm]
& & $R/\Rso$ \\
&  \\[-4mm]
\rotatebox{90}{\hspace{5.4cm} \rotatebox{-90}{(b)}} & \rotatebox{90}{\hspace{1.8cm} $\Weq(r)$} \hspace{-0.4cm} & \includegraphics[width=0.75\columnwidth]{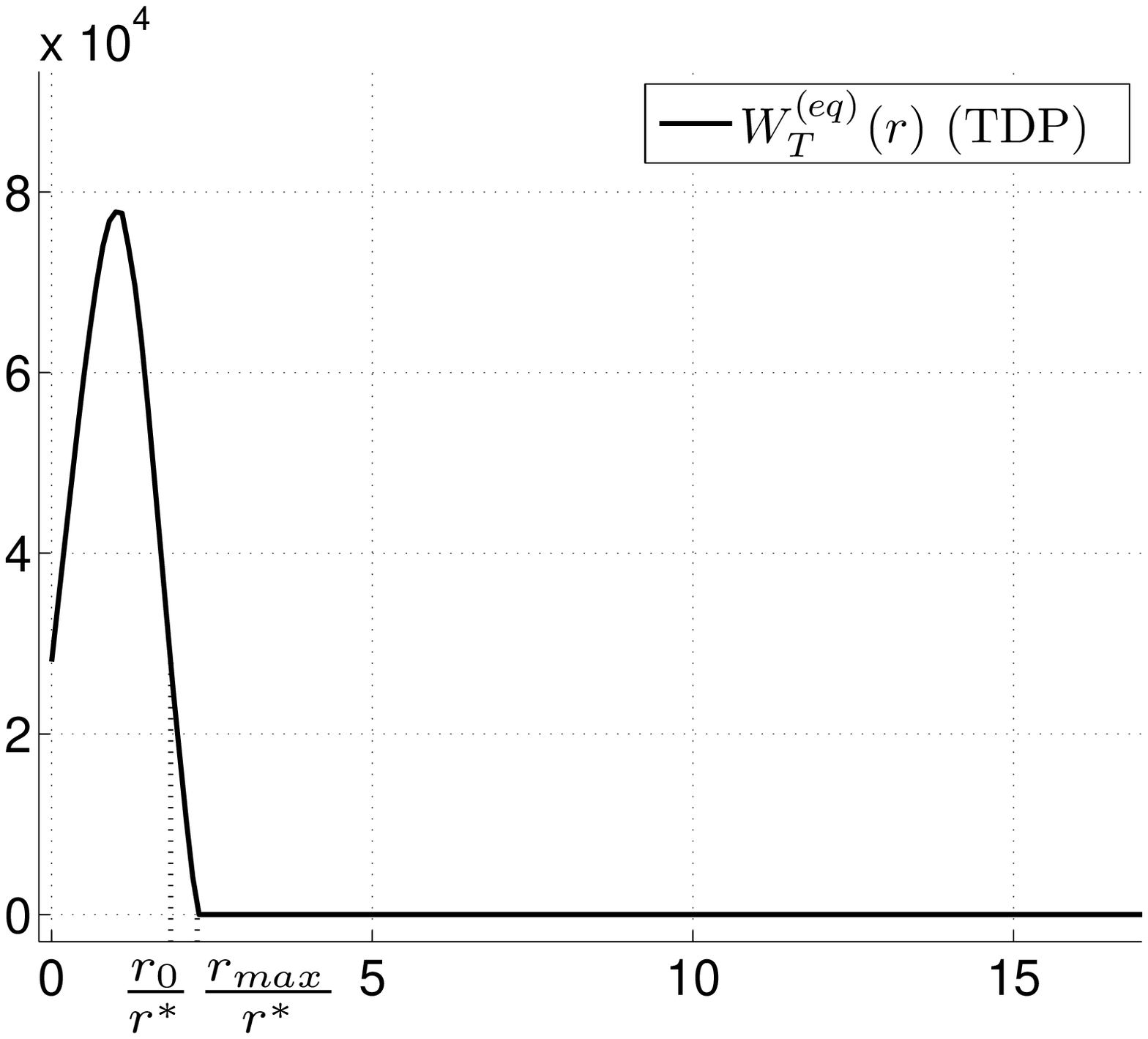} \\[-3mm]
& & $r/\rso$ 
\end{tabular}
\caption{Variation of the equilibrium aggregate welfare with the reward parameter for Example~\ref{ex.model_use_case}: (a) for the fixed-budget rebate mechanism -- (b) for the time-of-day pricing mechanism. }
\vspace{-5.5mm}
\label{fig.WrR}
\end{figure}

\subsection{Comparison of the two incentive mechanisms}
\label{sec.results2}

In this section, we compare the sensitivity of the two incentive mechanisms to imperfect information about user utilities. 
Let the games $\Gamma_{R} \left( \Theta, \mu, h, \{ c_{\theta}\}_{\theta\in\Theta}, \Rso \right)$ and $\Gamma_{T} \left( \Theta, \mu, h, \{ c_{\theta}\}_{\theta\in\Theta}, \rso \right)$ correspond to the baseline case of perfect information considered in the previous sections and suppose that $\Rso$ and $\rso$ have been chosen according to \eqref{eq.rRso} to induce a socially optimal level of public good at equilibrium (\ie $\Geq=\Gso$). We assume that $\Gso\in (0, D_p)$. 
We analyze the variations in equilibrium and in social optimum when $\Rso$ and $\rso$ are maintained for the respective mechanisms and utilities are perturbed (\ie actual utilities are different from the estimation used by the provider to set the parameters).

We restrict our analysis to the case where only the cost of shifting is perturbed and the rest of the utilities is unchanged. Indeed, we argue that it is more difficult to obtain data on the time preferences (the willingness to move demand from peak time to off-peak time) than on the total demand or on the sensitivity to delay. Therefore, the cost of shifting is more likely to be imperfectly estimated by the provider. We consider the following general form of the perturbed cost of shifting: \vspace{-1.5mm}
\begin{equation}
\label{eq.def_cti}
\cti_{\theta}(\cdot) = c_{\theta}(\cdot) + \epsilon \cdot p_{\theta}(\cdot),\\[-1mm]
\end{equation}
where $\epsilon$ is a real number and $p_{\theta}: [0, d_p] \to \R$ is a continuously differentiable function satisfying \vspace{-2mm}
\begin{equation*}
\sup_{\theta\in\Theta} \sup_{x\in[0, d_p]} |p^{\prime}_{\theta}(x)| < \infty. \\[-1.5mm]
\end{equation*}
Parameter $\epsilon$ is the perturbation magnitude and function $p_{\theta} (\cdot)$ is the direction of the perturbation. 
For the analysis, we restrict to small perturbations, \ie $\left|\epsilon\right|$ small. For $\left|\epsilon\right|$ small enough, the perturbed functions $\cti_{\theta}(\cdot)$ satisfy assumption {\bf (A2-3)}.
We assume that the perturbation direction is such that the aggregate best response has a non-zero perturbation at the order one in $\epsilon$ at the point $\Gso$ (the non-perturbed equilibrium). Otherwise, the equilibrium point would not be changed by the perturbation. 
For numerical illustrations, we will use the following simple perturbation which satisfies the above conditions: $p_{\theta}(\cdot) = c_{\theta}(\cdot)$ for all types $\theta\in\Theta$, \ie $c_{\theta}(\cdot)$ is scaled by a factor $1+\epsilon$ independent of the type.

Let $\Geq_R(\epsilon)$ and $\Geq_T(\epsilon)$ be  the equilibrium levels of public good in the games with perturbed utilities $\Gamma_{R} \left( \Theta,\mu,  h, \{ \cti_{\theta}\}_{\theta\in\Theta}, \Rso \right)$ and $\Gamma_{T} \left( \Theta,\mu,  h,  \{ \cti_{\theta}\}_{\theta\in\Theta}, \rso \right)$, respectively. Let $\Weq_R(\epsilon)$ and $\Weq_T(\epsilon)$ be the corresponding equilibrium welfares. 
Let $\Gso(\epsilon)$ and $\Wso(\epsilon)$ be the socially optimal level of public good with perturbed utilities, and the corresponding welfare resulting from the maximization of \eqref{eq.defW} where $c_{\theta}(\cdot)$ is replaced by $\cti_{\theta}(\cdot)$. 
To evaluate the variation of $\Geq$ with the perturbation, we need to evaluate the variation of the aggregate best response $\Gresp$ (recall that $\Geq$ is the  fixed-point of $\Gresp(\cdot)$, see \eqref{eq.fixed-point}). (The variation of $\Gso$ is handled similarly since from Theorem~\ref{thm.SO}, the social optimum can also be seen as a Nash equilibrium in a mechanism $SO$ with unit reward given by \eqref{eq.unit_r_SO}.) For this purpose, we introduce, for each mechanism $j\in\{R, T, SO\}$, the quantity $\alpha_{j}$ equal to the opposite of the slope of $\Gresp(\cdot)$ at the common non-perturbed equilibrium point $\Gso=\Gso(0)$ {\color{black}(see \eqref{eq.defalpha})}. We define the following conditions: \vspace{-1mm}
\begin{trivlist}
\item[\textbf{(C1)}] $\displaystyle \left|    \frac{1}{1+\alpha_{R}} - \frac{1}{1+\alpha_{SO}}   \right| < \left| \frac{1}{1+\alpha_{T}} - \frac{1}{1+\alpha_{SO}}   \right|$,\\[1pt]
\item[\textbf{(C2)}] $\displaystyle \left|    \frac{1}{1+\alpha_{R}} - \frac{1}{1+\alpha_{SO}}   \right| > \left| \frac{1}{1+\alpha_{T}} - \frac{1}{1+\alpha_{SO}}   \right|$.
\end{trivlist}
If the slopes $\alpha_{j}$'s for the different mechanisms are close enough, these conditions reduce to the following more intuitive  conditions (see details in Appendix~\ref{app.reduction}): \vspace{-1mm}
\begin{trivlist}
\item[\textbf{(C1$^\prime$)}] $\displaystyle \left|    r_R^{\prime}(G) - r_{SO}^{\prime}(G)   \right| < \left| r_T^{\prime}(G) - r_{SO}^{\prime}(G)   \right|$, at $G= \Gso(0)$, 
\item[\textbf{(C2$^\prime$)}] $\displaystyle \left|    r_R^{\prime}(G) - r_{SO}^{\prime}(G)   \right| > \left| r_T^{\prime}(G) - r_{SO}^{\prime}(G)   \right|$, at $G= \Gso(0)$, 
\end{trivlist}\vspace{-1mm}
where $r_R^{\prime}$,  $r_T^{\prime}$,  $r_{SO}^{\prime}$ are the respective derivatives of the unit rewards \vspace{-1mm}
\begin{subequations}
\label{eq.unit_r}
\begin{align}
r_R (G) &= \frac{R}{G}, \\[-.5mm]
r_T (G) &= r, \\
\label{eq.unit_r_SO} r_{SO} (G) &= \hp(G)(D-G).\\[-5.5mm] \nonumber
\end{align}
\end{subequations}

Then we have the following results.\vspace{-2mm}
\begin{proposition}%
\label{prop.robustness_h}
{\color{black} There exists $\epsilon_{m}>0$ such that, for any perturbation \eqref{eq.def_cti} with $\epsilon\neq0$ and $|\epsilon|<\epsilon_{m}$,}\vspace{-1mm}
\begin{trivlist}
\item[(i)] if condition \textbf{(C1)} is satisfied, then\vspace{-1mm}
$$\left| \Geq_R(\epsilon) - \Gso(\epsilon) \right| < \left| \Geq_T(\epsilon) - \Gso(\epsilon) \right|;\vspace{-1mm}$$
\item[(ii)] if condition \textbf{(C2)} is satisfied, then\vspace{-1mm}
$$\left| \Geq_R(\epsilon) - \Gso(\epsilon) \right| > \left| \Geq_T(\epsilon) - \Gso(\epsilon) \right|.\vspace{-1mm}$$
\end{trivlist}
\end{proposition}

The intuition behind Proposition~\ref{prop.robustness_h} is the following: the mechanism with the unit reward closer to the optimal unit reward $r_{SO}(G)$ have an equilibrium closer to the social optimum equilibrium $\Gso(\epsilon)$. 
Since $r_{R}(G)$ and $r_{SO}(G)$ are both decreasing functions, one expects $r_{R}(G)$ to be closer to $r_{SO}(G)$ than $r_{T}(G)$. It is often the case. The fact that $r_{R}(G)$ decreases when $G$ increases is the \emph{closed-loop} effect: the more users reduce their peak-time demand, the lower the incentive to reduce it is. 
However, if  $r_{R}(G)$ decreases much faster that  $r_{SO}(G)$, $r_{T}(G)$ can be closer to  $r_{SO}(G)$. This possibility is covered by case \textit{(ii)} of Proposition~\ref{prop.robustness_h}.

Fig.~\ref{fig.G_vs_gamma} illustrates the result of Proposition~\ref{prop.robustness_h} with the perturbation $p_{\theta}(\cdot) = c_{\theta}(\cdot)$ for all $\theta\in\Theta$. 
As it turns out, Example~\ref{ex.model_use_case} (Fig.~\ref{fig.G_vs_gamma}-(a)) falls in case \textit{(i)} of Proposition~\ref{prop.robustness_h} (see the unit rewards on Fig.~\ref{fig.unit_reward}); hence the fixed-budget rebate mechanism remains closer to social optimum than the time-of-day pricing mechanism. 
This is due to fact that the sensitivity to congestion is ``strongly convex'', \ie function $h(\cdot)$ \eqref{eq.defh} is far from linear. Hence, the optimal unit reward \eqref{eq.unit_r_SO} decreases ``fast'', as for the fixed-budget rebate mechanism. 
For the sole purpose of illustrating numerically  case \textit{(ii)} of Proposition~\ref{prop.robustness_h}, we construct the following example: 
\begin{example}
\label{ex.model_use_case_2}
Everything is defined as in  Example~\ref{ex.model_use_case}, but the disutility function is artificially contrived to have $h(G) = 1.2\cdot 10^{-3} \cdot (G^{0.95}-D_p^{0.95})$. (The factor $1.2\cdot 10^{-3}$ is chosen to yield the same social optimum level of public good $\Gso$ than in Example~\ref{ex.model_use_case} when $\epsilon=0$.)
\end{example}
Ex.~\ref{ex.model_use_case_2} is a contrived example where $h(\cdot)$ is almost linear so that the optimal unit reward $r_{SO}$ is almost constant, as in the time-of-day pricing mechanism. As a result, the time-of-day pricing mechanism is closer to social optimum (Fig.~\ref{fig.G_vs_gamma}-(b)).

Proposition~\ref{prop.robustness_h} compares the distance between $\Geq$ and $\Gso$ for the two incentive mechanisms, when utilities are perturbed. It holds for all perturbation directions $\{p_{\theta}(\cdot)\}_{\theta\in\Theta}$ (satisfying the conditions mentioned above). However, the direction of the variation of $\Geq$ and $\Gso$ depends on the perturbation directions $\{p_{\theta}(\cdot)\}_{\theta\in\Theta}$. With the simple perturbation where $p_{\theta}(\cdot) = c_{\theta}(\cdot)$ for all $\theta\in\Theta$, Fig.~\ref{fig.G_vs_gamma} shows that $\Geq$ and $\Gso$ decrease when $\epsilon$ increases (this could also be easily derived analytically from the proof of Proposition~\ref{prop.robustness_h}). Intuitively, if users are less willing to contribute due to a high cost of shifting, the equilibrium and social optimal amount of public good will be lower. 
With a general perturbation, the variation of $\Geq$ and $\Gso$ is determined by the variation of the aggregate best response $\Gresp$ at the point $\Gso(0)$ when the utilities perturbation is introduced.

\begin{figure}%
\centering
\begin{tabular}{ccc}
\rotatebox{90}{\hspace{4.2cm} \rotatebox{-90}{(a)}} & \rotatebox{90}{\hspace{1.4cm} $\Gso(\epsilon), \Geq (\epsilon)$} \hspace{-3mm} & \includegraphics[width=0.7\columnwidth]{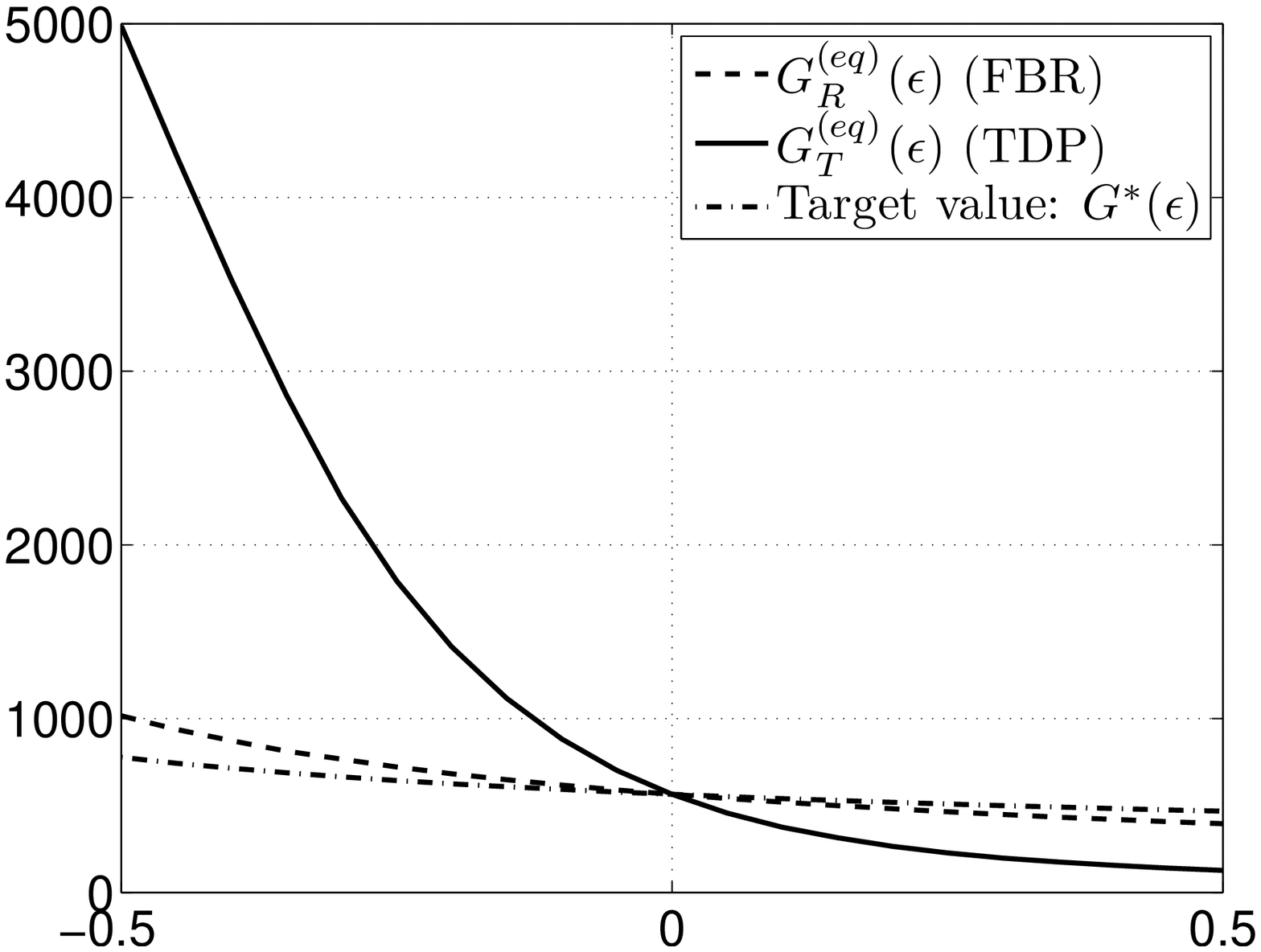} \\
& & \\[-4mm]
\rotatebox{90}{\hspace{4.2cm} \rotatebox{-90}{(b)}} & \rotatebox{90}{\hspace{1.4cm} $\Gso(\epsilon), \Geq (\epsilon)$} \hspace{-3mm} & \includegraphics[width=0.7\columnwidth]{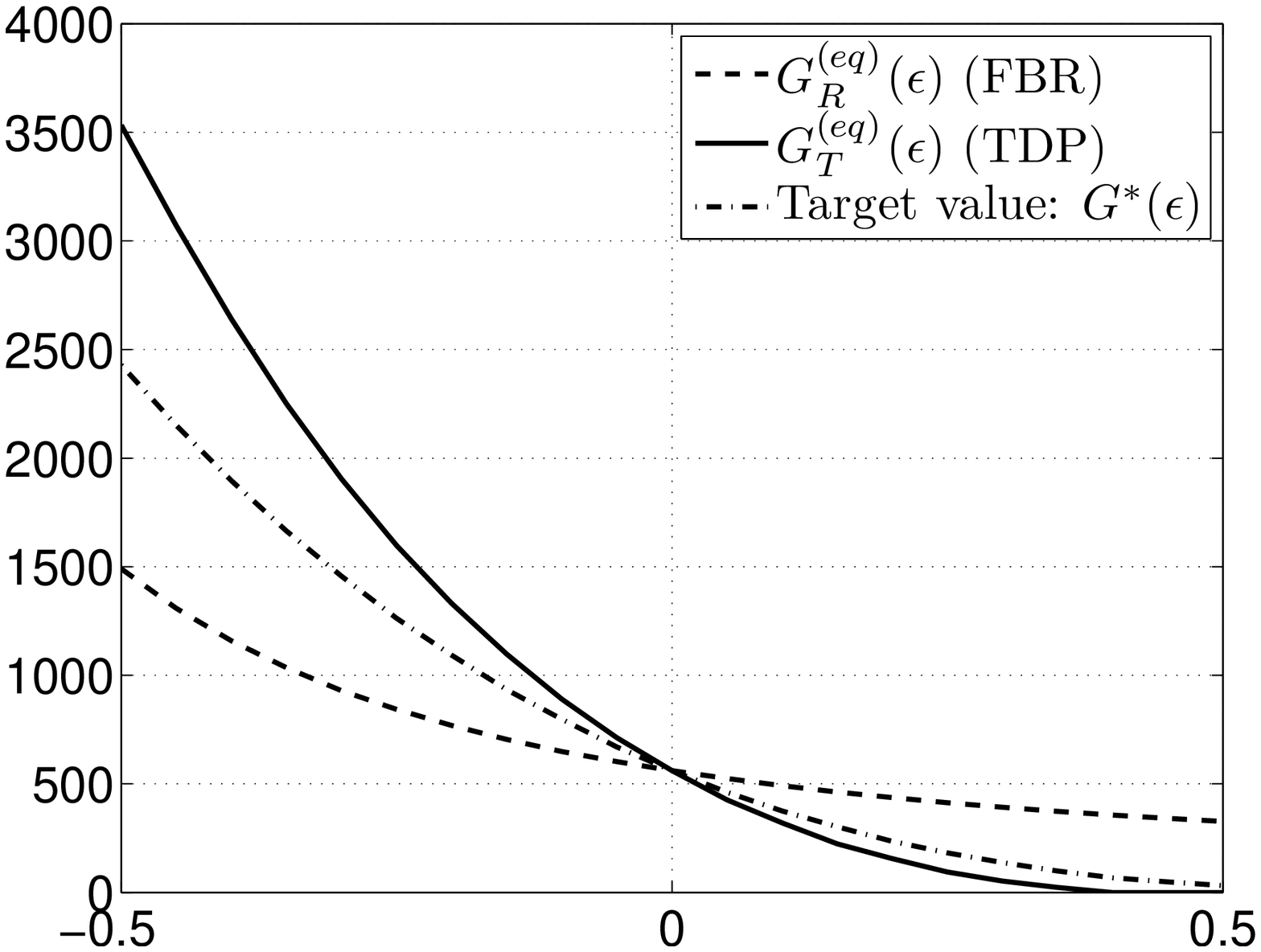} \\
& & $\epsilon$ 
\end{tabular}
\caption{
Variation of the equilibrium amount of public good $\Geq$ when functions $c_{\theta}(\cdot)$ are scaled by a factor $(1+\epsilon)$ starting from the baseline case: (a) for Example~\ref{ex.model_use_case} -- (b) for Example~\ref{ex.model_use_case_2}. 
}
\label{fig.G_vs_gamma}
\end{figure}

\begin{figure}%
\centering
\begin{tabular}{ccc}
\rotatebox{90}{\hspace{4.2cm} \rotatebox{-90}{(a)}} & \rotatebox{90}{\hspace{1.4cm} $r_{R}, r_{T}, r_{SO}$}  \hspace{-3mm}  & \includegraphics[width=0.65\columnwidth]{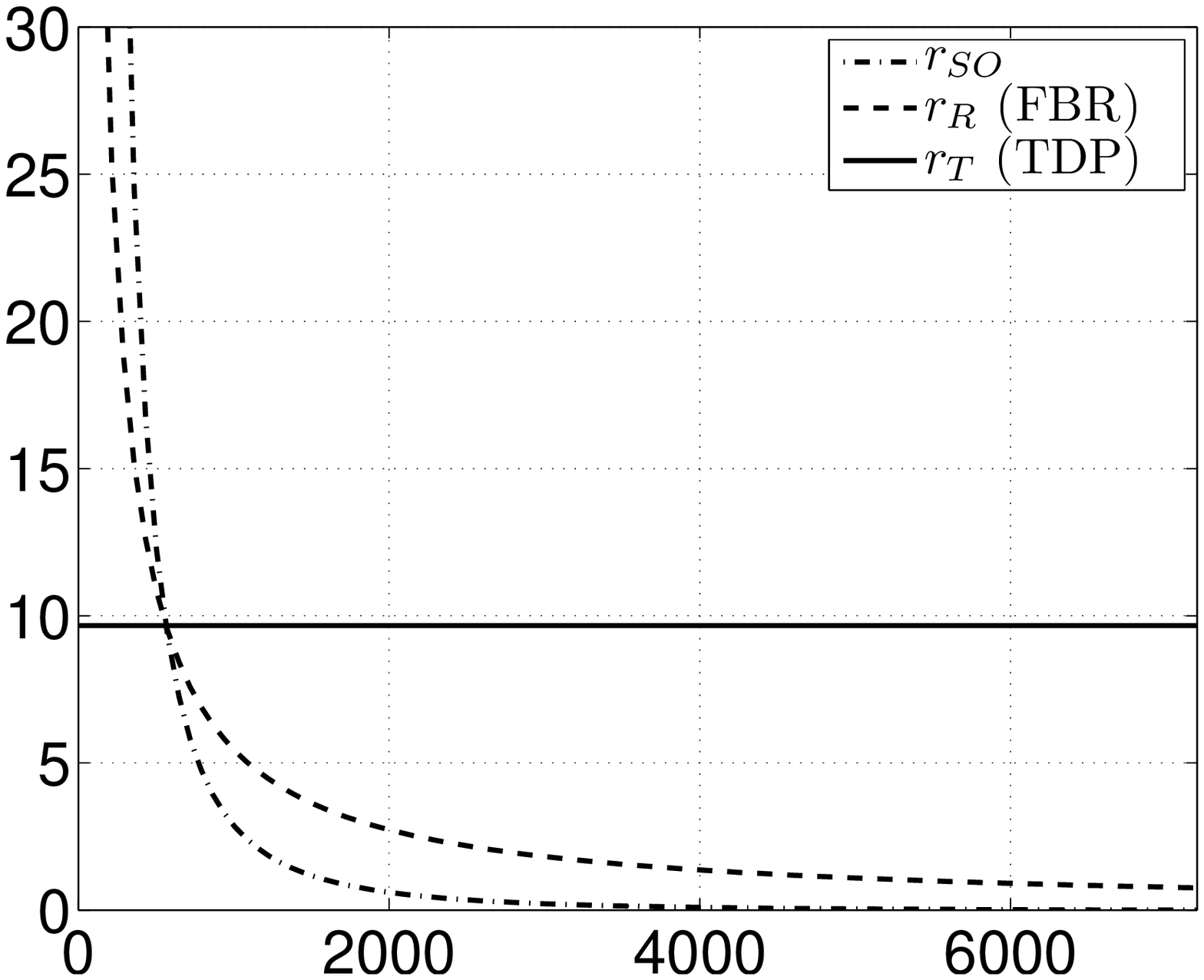} \\
& & \\[-4mm]
\rotatebox{90}{\hspace{4.2cm} \rotatebox{-90}{(b)}} & \rotatebox{90}{\hspace{1.4cm} $r_{R}, r_{T}, r_{SO}$}  \hspace{-3mm}  & \includegraphics[width=0.65\columnwidth]{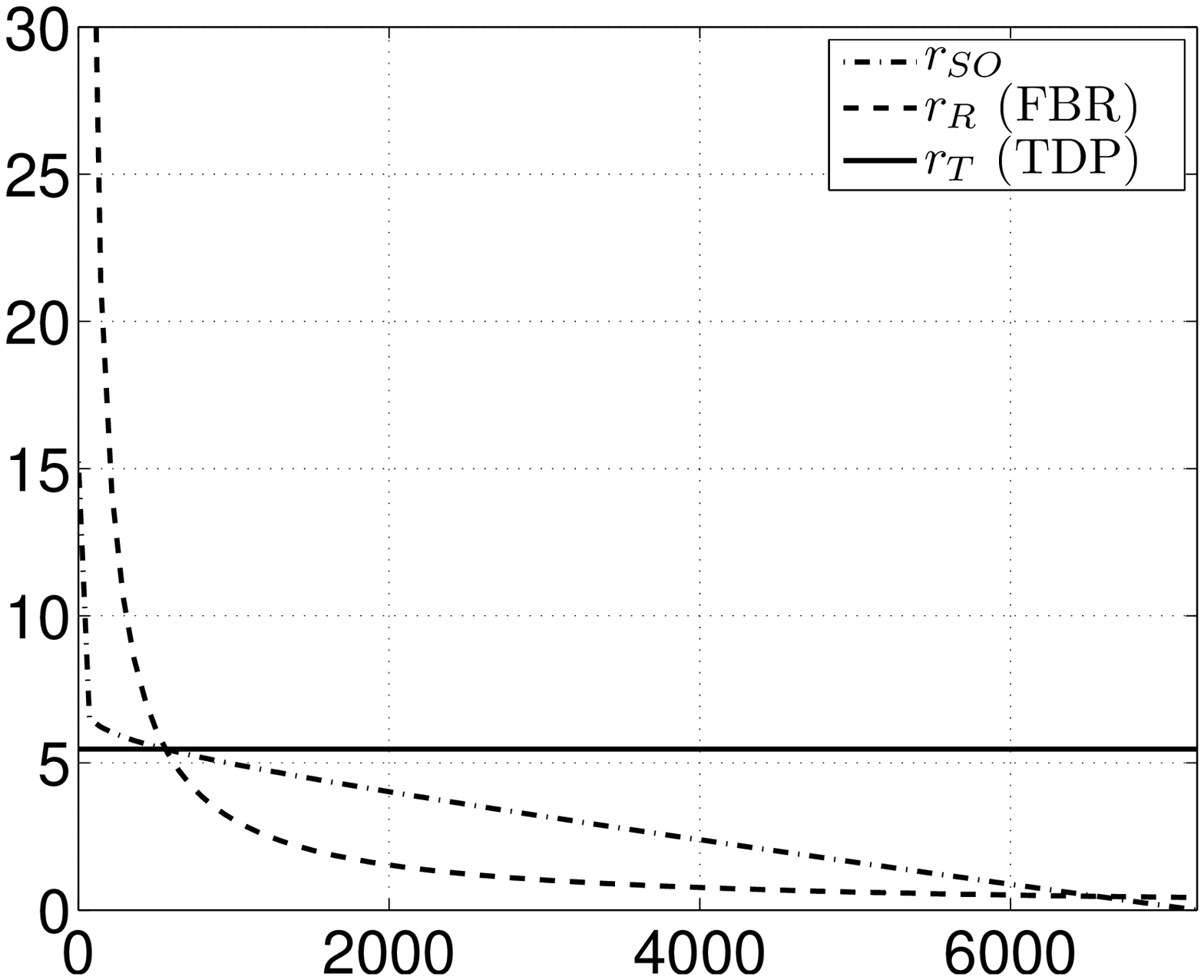} \\
& & $G$ 
\end{tabular}
\caption{
Unit reward as a function of $G$: (a) for Example~\ref{ex.model_use_case} -- (b) for Example~\ref{ex.model_use_case_2}. 
}
\label{fig.unit_reward}
\end{figure}

From Proposition~\ref{prop.robustness_h}, we deduce the following result.
\begin{theorem}%
\label{thm.robustness_h}
{\color{black} There exists $\epsilon_{m}>0$ such that, for any perturbation \eqref{eq.def_cti} with $\epsilon\neq0$ and $|\epsilon|<\epsilon_{m}$}, we have:\vspace{-.2mm}
\begin{enumerate}[(i)]
\item if condition \textbf{(C1)} is satisfied, then the fixed-budget rebate mechanism is more robust than the time-of-day pricing mechanism:\vspace{-2mm}
\begin{equation*}
\Weq_T(\epsilon) < \Weq_R(\epsilon) < \Wso(\epsilon);
\end{equation*}
\item if condition \textbf{(C2)} is satisfied, then the time-of-day pricing mechanism is more robust than the fixed-budget rebate mechanism:\vspace{-2mm}
\begin{equation*}
\Weq_R(\epsilon) < \Weq_T(\epsilon) < \Wso(\epsilon).
\end{equation*}
\end{enumerate}
\end{theorem}

Theorem~\ref{thm.robustness_h} is our main robustness result. It establishes which of the two mechanisms remains closer to optimal after the perturbation, in terms of welfare, \ie in terms of user expected utility over the timescale of a month (see Sec.~\ref{sec.measure}). The conditions of Theorem~\ref{thm.robustness_h} are the same as in Proposition~\ref{prop.robustness_h}: mechanism $j\in \{ R,T \}$ is more robust if its unit reward is closer to the optimal unit reward. Since Example~\ref{ex.model_use_case} satisfies condition \textbf{(C1)} (due to ``strong enough'' convexity of the sensitivity to congestion), the fixed-budget rebate mechanism is more robust. It means that if the provider chooses the parameters based on an imperfect estimation of the cost of shifting $c_{\theta}(\cdot)$, the welfare of a population of users whose actual cost of shifting is $\cti_{\theta}(\cdot)$ will be higher with the fixed-budget rebate mechanism than with the time-of-day pricing mechanism. 
Similarly, if the cost of shifting was varying according to a given probability law and parameters $R$ and $r$ were chosen based on expectations, the fixed-budget rebate mechanism would give a higher expected welfare.

Fig.~\ref{fig.W_vs_gamma} illustrates our robustness results for a simple perturbation. It shows that our analysis with $\epsilon$ close to zero extends to larger perturbations. The numerical values for Example~\ref{ex.model_use_case} are reported in Tab.~\ref{tab.2}: for a $50\%$ error in the cost-of-shifting estimation ($\epsilon=0.5$), the welfare is $20\%$ below optimal with the time-of-day pricing mechanism, whereas it is only $0.3\%$ below optimal with the fixed-budget rebate mechanism.

\begin{table}
\centering
\caption{Equilibrium for Example~\ref{ex.model_use_case}, with utilities perturbed by a scaling of factor $(1+\epsilon)$ with $\epsilon=0.5$. The parameters are chosen based on unperturbed utilities: $R=\$5,500$ and $r=\$9$/Gbit (see Fig.~\ref{fig.case1}). }
\vspace{-2mm}
\label{tab.2}
\begin{tabular}{|c|c|c|c|c|}
\hline
& no mecha. & TDP mecha. & FBR pricing & SO  \\
\hline
\hline
G & 37~Gbits & 127~Gbits & 395~Gbits & 467~Gbits \\
W & $\sim$0 & 61,000 & 75,200 & 75,400 \\
\hline
\end{tabular}
\vspace{-5mm}
\end{table}

\begin{figure}%
\centering
\begin{tabular}{ccc}
\rotatebox{90}{\hspace{4.2cm} \rotatebox{-90}{(a)}} & \rotatebox{90}{\hspace{1.1cm} $\Wso(\epsilon)-\Weq(\epsilon)$} \hspace{-3mm} & \includegraphics[width=0.68\columnwidth]{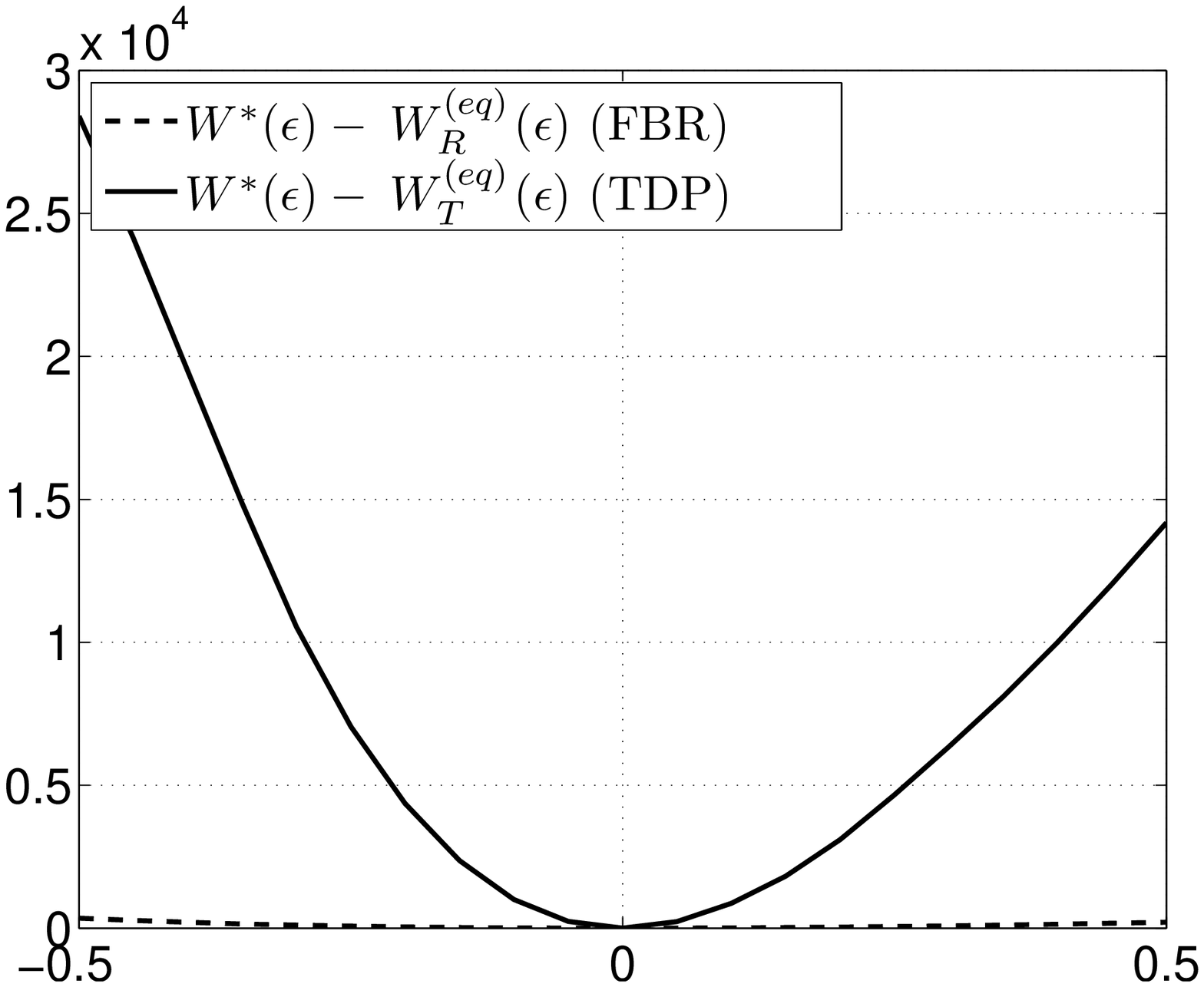} \\
& & \\[-4mm] 
\rotatebox{90}{\hspace{4.2cm} \rotatebox{-90}{(b)}} & \rotatebox{90}{\hspace{1.1cm} $\Wso(\epsilon)-\Weq(\epsilon)$} \hspace{-3mm} & \includegraphics[width=0.7\columnwidth]{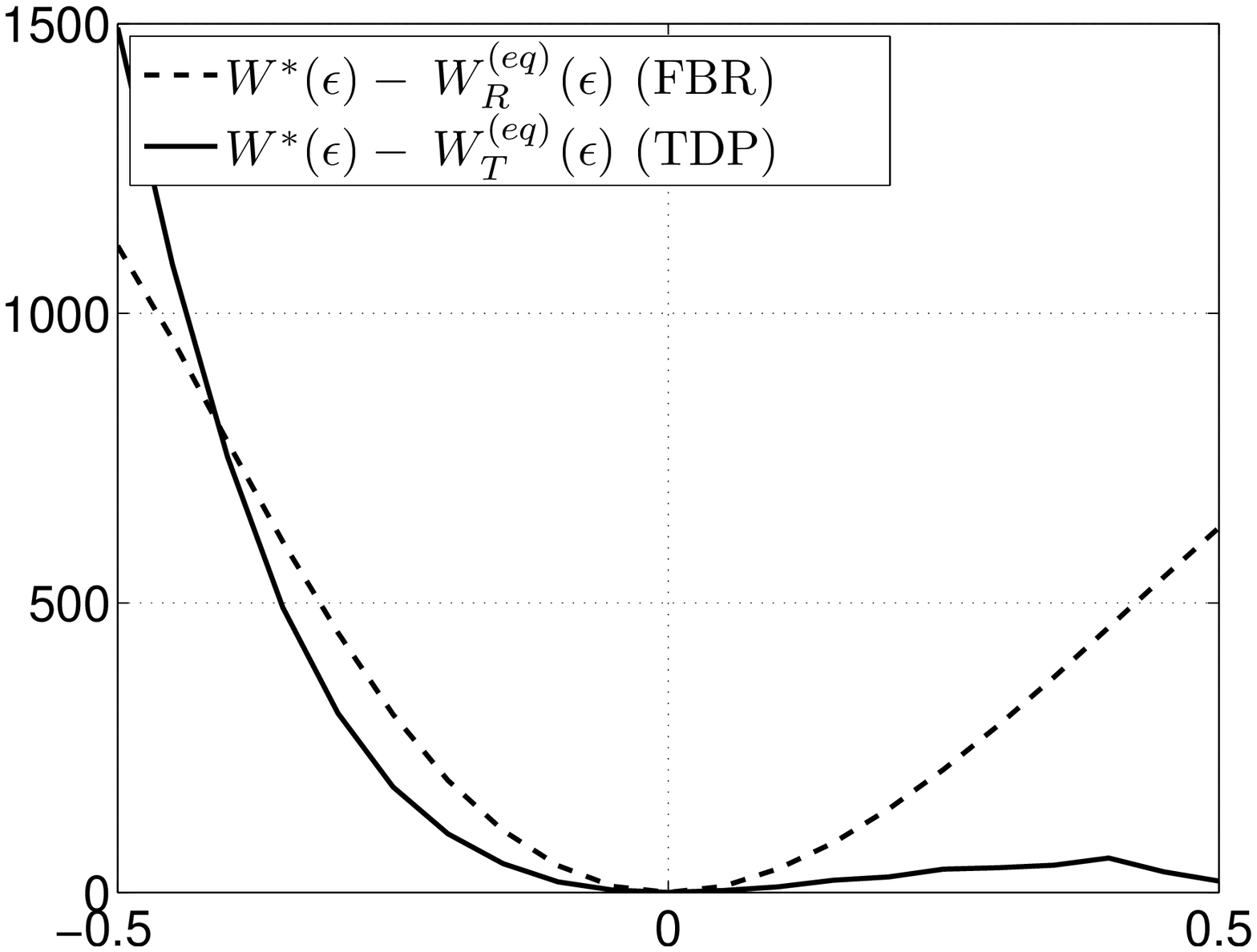} \\
& & $\epsilon$ 
\end{tabular}
\caption{
Variation of the aggregate welfare $\Weq$ when functions $c_{\theta}(\cdot)$ are scaled by a factor $(1+\epsilon)$ starting from the baseline case: (a) for Example~\ref{ex.model_use_case} -- (b) for Example~\ref{ex.model_use_case_2}. For readability and robustness comparison, the difference $\Wso(\epsilon)-\Weq(\epsilon)$ is plotted. 
}
\vspace{-5mm}
\label{fig.W_vs_gamma}
\end{figure}

\section{Concluding remarks}
\label{sec.conclusion}

This paper provides a comparative analysis of two incentive mechanisms of reducing peak-time congestion in Internet broadband access: a fixed-budget rebate mechanism inspired by the economic literature on public good provision by means of lotteries, and a more standard for the network literature time-of-day pricing mechanism. 
The fixed-budget rebate mechanism can be interpreted as probabilistic pricing: in this mechanism, each user's reward depends not only on his contribution but also on the contribution of the other users.  %
We suggest that this mechanism has two advantages relative to the time-of-day pricing mechanism with given prices for specific time slots. Firstly, the fixed-budget rebate mechanism is easy to implement via lottery-like scheme(s), for which a total user reward is announced by the ISP in advance. Secondly, it has built-in self-tuning, which appears to be attractive in environments with imperfectly known demand.

{\color{black} Our paper uses a simplified model to provide a theoretical structure that permits to understand the benefits of the fixed-budget rebate mechanism over more standard approaches. The deployment of the mechanism will raise a number of practical questions. In particular, the ISP has to decide at which scale to deploy the mechanism: deploying it at the scale of a base station would involve too precise monitoring whereas deploying it region-wide would face the issue that users do not all share the same access bottleneck. We believe that such decision should be made based on historical statistics on each bottleneck that are  accessible to ISPs.  } 
{\color{black} Our model also considers only two time periods whereas it could be useful for an ISP to use a finer subdivision of the day. Again, we believe that the number of time periods should be determined using historical data available to the provider. The fixed-budget rebate mechanism could be easily extended to multiple time periods and we believe that it would remain more robust than the time-of-day pricing mechanism.  }

In our model, we considered a monopolist and the reward was financed by an increase in the subscription price. We showed that both mechanisms still improve each user's average utility. In different scenarios, if the subscription price cannot be increased, it would be possible to finance the reward by a different means, \eg by the reduction of the congestion cost or by the higher number of customers that the provider could accommodate with the same infrastructure due to lower congestion. 

{\color{black}  Our model focuses on user welfare maximization rather than on the cost savings for the ISP. However, we believe that both objectives are consistent, as in an competitive environment, increasing user welfare allows either to accommodate more users with the same capacity or to reduce the capacity provisioning costs for the same user base. A quantitative analysis of these questions would require modeling of the cost structure and of the competition (we could typically assume perfect competition) and is left as future work.   }

{\color{black} To implement in practice the fixed-budget rebate mechanism, an ISP will need to track separately the consumption at peak and off-peak time. However, such separate accounting is also needed for the time-of-day pricing mechanism and is already technically possible in most settings (mobile access, electricity with smart meters, etc.). The use of probabilistic pricing also raises the question of contention billing and verifiability. However, these transparency issues are the same as in many other contexts where they have been successfully solved (state lotteries, casinos, etc.). For instance, inspection techniques similar to those used for gambling \cite{Gaming} could be used here. }

While our motivating application is telecommunications, we believe that the fixed-budget rebate mechanism could be modified for use in other applications such as electricity and transportation networks. 
In the case of electricity demand management, privacy and security considerations make our mechanism advantageous relative to real-time pricing. %
Indeed, our mechanism requires no real-time user-dispatcher communication. 
In addition, unlike currently suggested real-time pricing mechanisms (\eg \cite{LeBoudec11a}), our mechanism requires only aggregate data. 
We plan to explore these other applications in future work.

\bibliographystyle{IEEEtran}
\bibliography{../../biblio_lotteries.bib}

\appendices

\section{Properties of the cost of shifting}
\label{app.costconvex}

Let $\theta\in\Theta$. 
First, we show that the function $x_{\theta} \mapsto c_{\theta} (x_{\theta})$, defined by \eqref{eq.cost}, is differentiable on $[0, d_p]$. By assumption, the function $x_{\theta}\mapsto P_{\theta}(d_p-x_{\theta})-(d_p-x_{\theta})q$ is differentiable so we only need to show that the function $x_{\theta}\mapsto O_{\theta}(z^{*}_{\theta} (x_{\theta}))-z^{*}_{\theta} (x_{\theta})q$ at $z_{\theta, max}$ is differentiable on $[0, d_p]$.  Since the function $z\mapsto O_{\theta}(z)$ is twice differentiable increasing strictly concave by assumption, the function $z\mapsto O_{\theta}(z)-zq$ is twice differentiable strictly concave. Let $z_{\theta, max}$ be its maximum on $[0, d_p]$. Then, \vspace{-3mm}
\begin{equation*}
z^{*}_{\theta} (x_{\theta}) = \left\{ 
\begin{array}{l}
x_{\theta} \textrm{ if } x_{\theta} \leq z_{\theta, max}, \\ 
z_{\theta, max} \textrm{ if } x_{\theta} > z_{\theta, max}. 
\end{array}
\right.\\[-1mm]
\end{equation*}
If $z_{\theta, max} = 0$ (hence $z^{*}_{\theta} (x_{\theta}) = z_{\theta, max}, \forall x \in [0, d_p]$), or if $z_{\theta, max} = d_p$ (hence $z^{*}_{\theta} (x_{\theta}) = x_{\theta}, \forall x \in [0, d_p]$), then $x_{\theta}\mapsto O_{\theta}(z^{*}_{\theta} (x_{\theta}))-z^{*}_{\theta} (x_{\theta})q$ is clearly differentiable on $[0, d_p]$. 
If $z_{\theta, max} \in (0, d_p)$, we need to show that $x_{\theta}\mapsto O_{\theta}(z^{*}_{\theta} (x_{\theta}))-z^{*}_{\theta} (x_{\theta})q$ is differentiable at the point $z_{\theta, max}$. But since $z_{\theta, max}$ is an interior maximum, we have $O_{\theta}(z_{\theta, max})-q=0$, therefore both the left and right derivatives of $x_{\theta}\mapsto O_{\theta}(z^{*}_{\theta} (x_{\theta}))-z^{*}_{\theta} (x_{\theta})q$ at $z_{\theta, max}$ are zero.

Next, to show that the function $x_{\theta} \mapsto c_{\theta} (x_{\theta})$ is strictly convex on $[0, d_p]$, we show that its derivative is increasing. By assumption that $P_{\theta}(\cdot)$ is strictly concave, the derivative of $x_{\theta}\mapsto -\left[ P_{\theta}(d_p-x_{\theta})-(d_p-x_{\theta})q\right]$ is increasing. 
Moreover, the derivative of $x_{\theta}\mapsto -\left[O_{\theta}(z^{*}_{\theta} (x_{\theta}))-z^{*}_{\theta} (x_{\theta})q\right]$ is  increasing on $[0, z_{\theta, max})$ and constant on $[z_{\theta, max}, d_p]$. Therefore the derivative of  $x_{\theta} \mapsto c_{\theta} (x_{\theta})$ is increasing. 

Last, we show that $x_{\theta} \mapsto c_{\theta} (x_{\theta})$ is increasing on $[\xl_{\theta}, d_p]$. If $\xl_{\theta}=d_p$, the result is trivial. If $\xl_{\theta}<d_p$, by definition of $\xl_{\theta}$,  $\cpr_{\theta}(\xl_{\theta}) \ge 0$. Hence, since we have shown that $\cpr_{\theta}(\cdot)$ is increasing, we have $\cpr_{\theta}(x_{\theta}) > 0, \forall x\in(\xl_{\theta}, d_p]$ which gives the result.

\section{Proof of \protect Theorem~\ref{thm.Nash}}
\label{app.proof_existenceuniqueness}

Let $R$ and $r$ be fixed and let $j\in\{R, T\}$. For a given $G\in[0, D_p]$, the best response (solving the FOCs~\eqref{eq.FOC_continuous}) defines a measurable function $\xresp(G): \Theta\to [0, d_p]$ given for all $\theta\in\Theta$ by \vspace{-4mm}
\begin{equation}
\label{eq.xresp}
\xresp_{\theta}(G) \!\!=\!\! \left\{
\begin{array}{l}
 0,   \textrm{ if }   {M^j}^{\prime}(G) - h(G) - \dpr_{\theta}(0) \le 0, \\
 d_p,   \textrm{ if }   {M^j}^{\prime}(G) - h(G) - \dpr_{\theta}(d_p) \ge 0,\\ %
 (\dpr_{\theta})^{-1} \left(  {M^j}^{\prime}(G) - h(G)\right),  \textrm{ otherwise}.%
\end{array}
\right.\\[-1.5mm]
\end{equation}
Due to assumption \textbf{(A2)}, $\dpr_{\theta}(\cdot)$ is strictly increasing, hence invertible and with an increasing inverse function. Therefore, \eqref{eq.xresp} uniquely defines $\xresp(G)$. %
Let \vspace{-1.5mm}
\begin{equation}
\label{eq.fixed-point2} \Gresp(G) = \int_{\Theta}\xresp_{\theta}(G) \ud \mu(\theta) \\[-1.5mm]
\end{equation}
be the aggregate best response. %
By definition and strict concavity of the utility function,  a measurable function $x: \Theta\to [0, d_p]$ is a Nash equilibrium if and only if there exists $G\in[0, D_p]$ satisfying the fixed-point equation~\eqref{eq.fixed-point} such that $x=\xresp(G)$. 
To conclude the proof, we show that \eqref{eq.fixed-point} admits a unique fixed-point (see illustration on Fig.~\ref{fig.case1}). \vspace{-1mm}
\begin{lemma}%
\label{lem.unique_fp}
There exists a unique solution of \eqref{eq.fixed-point}. 
\end{lemma}
\vspace{-2mm}
\begin{proof}
The r.h.s. of \eqref{eq.fixed-point} ($G$) is clearly a strictly increasing continuous function of $G$,  from $[0, D_p]$ to $[0, D_p]$. 

For the l.h.s. ($\Gresp(G)$), firstly note that it is a continuous function of $G$. Indeed, due to assumption \textbf{(A2)},  $(\cpr_{\theta}(\cdot))^{-1}(\cdot)$ is strictly increasing continuous, hence $\xresp_{\theta}(G)$ is continuous in $G$  for all $\theta\in\Theta$. Moreover, the function $\xresp(G): \Theta\to [0, d_p]$ is dominated by the constant function equal to $d_p$ (\ie $| \xresp(G)|\le d_p$) which is integrable w.r.t. $\mu$. Therefore, for any $G\in [0, D_p]$ and for any sequence $(G_n)_{n\ge0}$ which converges to $G$, we have $\xresp(G_n)\xrightarrow[n\to\infty]{}\xresp(G)$ pointwise (by continuity of $\xresp_{\theta}(G)$ w.r.t to $G$ for all $\theta\in\Theta$) and by Lebesgue dominated convergence theorem, \vspace{-1mm}
\begin{eqnarray}
\nonumber \lim_{n\to\infty} \Gresp(G_n) &=& \lim_{n\to\infty}\int_{\Theta}\xresp_{\theta}(G_n) \ud \mu(\theta) \\ 
\nonumber  &=& \int_{\Theta} \left[ \lim_{n\to\infty} \xresp_{\theta}(G_n) \right] \ud \mu(\theta) \\ 
\nonumber &=& \Gresp(G). \\[-6mm]\nonumber
\end{eqnarray}

Clearly, the l.h.s. ($\Gresp(G)$) is also a non-increasing function of $G$ taking values in $[0, D_p]$. Therefore, there is a unique fixed-point of \eqref{eq.fixed-point}.
\end{proof}

\section{Proof of \protect Theorem~\ref{thm.SO}}
\label{app.proof_SO}

We first prove \textit{(i)}. Let $\X$ be the set of functions $x: \Theta\to\R$ such that $\int_{\Theta} \left| x_{\theta} \right| \ud \mu( \theta) < \infty$ and  let $\X_0\subset\X$ be the set of functions $x: \Theta\to [0, d_p]$. Consider the aggregate welfare \eqref{eq.defW} as a functional on $\X_0$ taking values in $\R$: \vspace{-1.5mm}%
\begin{eqnarray}
\nonumber W (x) &=& h\left(\int_{\Theta} x_{\theta} \ud \mu ( \theta)\right) \left(D_p - \int_{\Theta} x_{\theta} \ud \mu ( \theta)\right) \\ 
\nonumber && -  \int_{\Theta} c_{\theta} (x_{\theta}) \ud \mu ( \theta) + \int_{\Theta} \bar{u}_{\theta} \ud \mu ( \theta)  - p\frac{D_p}{d_p}.
\end{eqnarray}

\vspace{-1.5mm}
Since $\X_0$ is compact and the functional $W$ is continuous, it has a maximum (see Corollary 38.10 of \cite[p. 152]{Zeidler85a}). 
Let $\xso\in \X_0$ be such that $W$ is maximal and let \vspace{-2mm}
$$\Gso = \int_{\Theta} \xso_{\theta} \ud \mu(\theta).\\[-1mm]$$ 
Define the three subsets of $\Theta$: $\Bp_1$, $\Bp_2$ and $\Bp_3$ where $\xso=0$, $\xso\in (0, d_p)$ and $\xso=d_p$ respectively. 
We now derive necessary conditions for $\xso$ to maximize $W$ in each subset. %

We start with the subset $\Bp_2$ corresponding to interior points. Let $y\in\X$ be such that $y_{\theta}=0$ for all $\theta\in\Theta\backslash\Bp_2$. We define the directional derivative (also called G\^ateaux derivative) of $W$ around $\xso$ in the direction $y$ as\vspace{-2mm}
\begin{equation*} 
 \ud W (\xso,y) = \lim_{t\to 0} \frac{W (\xso+ty) - W (\xso)}{t}.\\[-2mm]
\end{equation*}
Then, we have\vspace{-2mm}
\begin{align*} 
 \ud W (\xso\!\!,y) \!=\!\!\!
\int_{\Theta_2} \!\!\!\!y_{\theta} \! \cdot \! \left[ \hp(\Gso)(D_p-\Gso) \!-\! h(\Gso) \!-\! \dpr_{\theta}(\xso_{\theta}) \right] \ud \mu (\theta),
\end{align*}
where the exchange between limit and integration in the last term (giving $- y_{\theta} \dpr_{\theta}(x_{\theta})$) is justified by Lebesgue's dominated convergence theorem whenever $\int_{\Theta} \left| y_{\theta} \cdot \dpr_{\theta}(x_{\theta}) \right| \ud \mu ( \theta) < \infty$. This holds here due to assumption \textbf{(A3)}.  \\
For $\xso$ to be optimal, it is necessary that $\ud W (\xso,y)=0$, \ie\vspace{-2mm}
\begin{equation*}
 \int_{\Bp_2} y_{\theta} \cdot \left[ \hp(\Gso)(D_p-\Gso) - h(\Gso) - \dpr_{\theta}(\xso_{\theta}) \right] \ud \mu ( \theta) = 0.\\[-1mm]
\end{equation*}
For this to hold for any function $y$ such that $y_{\theta}=0$ for all $\theta\in\Theta\backslash\Bp_2$, it is necessary that we have\vspace{-2mm}
\begin{equation}
\label{eq.nec1}
\hp(\Gso)(D_p-\Gso) - h(\Gso) - \dpr_{\theta}(\xso_{\theta}) = 0,\\[-1mm]
\end{equation}
for almost-all $\theta\in\Bp_2$, \ie for almost-all $\theta$ such that $\xso_{\theta}\in (0, d_p)$.

We now treat the case of subset $\Bp_1$, which corresponds the points of the lower boundary. Let $y\in\X$ be such that $y_{\theta}\ge0$ for all $\theta\in\Bp_1$ and $y_{\theta}=0$ for all $\theta\in\Theta\backslash\Bp_1$; that is $y$ is a direction that ``pushes up'' the values of $\xso$ that are at zero. The directional derivative of $W$ around $\xso$ in the direction $y$ is defined similarly to the previous case but with a limit $t>0$:\vspace{-2mm}
\begin{equation*} 
 \ud W (\xso,y) = \lim_{t\to 0^+} \frac{W (\xso+ty) - W (\xso)}{t},\\[-2mm]
\end{equation*}
which gives\vspace{-2mm}
\begin{align*} 
 \ud W (\xso\!\!,y) \!=\!\!\!
\int_{\Theta_1} \!\!\!\!y_{\theta} \! \cdot \! \left[ \hp(\Gso)(D_p-\Gso) \!-\! h(\Gso) \!-\! \dpr_{\theta}(\xso_{\theta}) \right] \ud \mu (\theta).\\[-6mm]
\end{align*}
For $\xso$ to be optimal, it is necessary that $\ud W (\xso,y)\le0$, \ie\vspace{-2mm}
\begin{equation*}
 \int_{\Bp_1} y_{\theta} \cdot \left[ \hp(\Gso)(D_p-\Gso) - h(\Gso) - \dpr_{\theta}(\xso_{\theta}) \right] \ud \mu ( \theta) \le 0.\\[-1mm]
\end{equation*}
For this to hold for any function $y$ such that  $y_{\theta}\ge0$ for all $\theta\in\Bp_1$ and $y_{\theta}=0$ for all $\theta\in\Theta\backslash\Bp_1$, it is necessary that \vspace{-2mm}
\begin{equation}
\label{eq.nec2}
\hp(\Gso)(D_p-\Gso) - h(\Gso) - \dpr_{\theta}(\xso_{\theta}) \le 0,
\end{equation}
for almost-all $\theta\in\Bp_1$, that is for almost all $\theta$ such that $\xso_{\theta}=0$.

The case of subset $\Bp_3$ is handled similarly and yields the necessary condition: \vspace{-2mm}
\begin{equation}
\label{eq.nec3}
\hp(\Gso)(D_p-\Gso) - h(\Gso) - \dpr_{\theta}(\xso_{\theta}) \ge 0,\\[-1.5mm]
\end{equation}
for almost all $\theta$ such that $\xso_{\theta}=d_p$.

In summary, \eqref{eq.nec1}-\eqref{eq.nec3} show that for function $\xso$ to maximize $W$, it is necessary that $\xso$ is solution of the FOCs~\eqref{eq.FOC_continuous} where ${M^j}^{\prime}(G)$ is replaced by $\hp(G)(D-G)$. By assumption~\textbf{(A1)}, this is a decreasing function of $G$. Therefore the same proof as for Theorem~\ref{thm.Nash} shows that $\xso$ is uniquely determined almost-everywhere.

We now prove \textit{(ii)}. From the proof of part \textit{(i)}, it is clear that if $R=\Rso$ (resp. $r=\rso$), then the FOCs~\eqref{eq.FOC_continuous} at a Nash equilibrium coincide with the optimality condition, which gives the result.

\section{Proof of \protect Proposition~\ref{prop.variation_rR}}
\label{app.proof_variation_rR}

We provide the proof for the fixed-budget rebate mechanism ($j=R$). It is the same for the time-of-day pricing mechanism.

We first prove \textit{(i)} using three cases. 

{\it Case 1:} If $\Geq(R) = 0$, then the result is obvious. %

{\it Case 2:} If $\Geq(R) = D_p$, then we have $R/D_p - h(D_p) - \dpr_{\theta}(d_p) \ge 0$ for almost-all $\theta\in\Theta$, which implies  $\Rp/D_p - h(D_p) - \dpr_{\theta}(d_p) \ge 0$. Hence $\Geq(\Rp) = D_p$.

{\it Case 3:} If $\Geq(R) \in (0, D_p)$. For a given $G$, $\xresp(G)$ of \eqref{eq.xresp} is non-decreasing when $R$ increases to $\Rp$, and strictly increasing for $\theta$'s s.t. $\xresp(\Geq(R))\in(0, D_p)$. Since the set of such $\theta$'s is of positive measure, the new fixed-point has $\Geq(\Rp) > \Geq(R)$.

From \textit{(i)}, \textit{(ii)} follows clearly.

We finally prove \textit{(iii)}. 
The existence of the threshold $\Rbar$ follows from the fact that $ {M^R}^{\prime}(D_p) \xrightarrow[R\to\infty]{} \infty$, hence ${M^R}^{\prime}(D_p) - h(D_p) - \dpr_{\theta}(d_p) \ge 0$ for almost-all $\theta\in\Theta$ beyond $\Rbar$ (due to assumption~\textbf{(A3)}).

\section{Proof of \protect Proposition~\ref{prop.raffle}}
\label{app.proof_raffle}

For the fixed-budget rebate mechanism with $R>0$, if $G=0$, the unit reward is infinite, hence each user wants to contribute positively. Therefore, $G=0$ is not an equilibrium.

\section{Proof of \protect Proposition~\ref{prop.W}}
\label{app.proof_propW}

We provide the proof for the fixed-budget rebate mechanism ($j=R$). It is the same for the time-of-day pricing mechanism.

Let $R\ge0$ and denote for simplicity $x = \xeq(R)$ and $G = \Geq(R)$. Using the notation of the proof of Theorem~\ref{thm.SO} (Appendix~\ref{app.proof_SO}), the derivative of the welfare around $x$ in the direction $y\in \X$ is \vspace{-2mm}
\begin{equation*} 
\ud W (x,y) \!=\!\! \int_{\Theta} y_{\theta} \cdot \left[ \hp(G)(D_p-G) - h(G) - \dpr_{\theta}(x_{\theta}) \right] \ud \mu ( \theta).\\[-1mm]
\end{equation*}

Suppose first that $R<\Rso$ and consider the direction $y\in\X$ such that \vspace{-2mm}
$$y_{\theta} = \dd{\xeq_{\theta}(R)}{R}, \quad \forall \theta\in\Theta.\\[-1mm]$$ 
By Proposition~\ref{prop.variation_rR}, we have $\Geq(R)\le\Gso$ and $\xeq_{\theta}(R) \le \xso_{\theta}$ for almost-all $\theta\in\Theta$. 
Therefore, we have $y_{\theta} \ge 0$ for all $\theta\in\Theta_2\cup\Theta_3$ and $y_{\theta} = 0$ for $\theta\in\Theta_1$ (recall that $\Theta_1$, $\Theta_2$ and $\Theta_3$ are the subsets of $\Theta$ where $\xso=0$, $\xso\in (0, d_p)$ and $\xso=d_p$ respectively). 
Moreover, we have $\hp(G)(D_p-G) - h(G) - \dpr_{\theta}(x_{\theta}) \ge 0$ for all $\theta\in\Theta_2\cup\Theta_3$. Indeed, in $\Theta_2$, $\hp(\Gso)(D_p-\Gso) - h(\Gso) - \dpr_{\theta}(\xso_{\theta}) = 0$ and $\hp(G)(D_p-G) - h(G) \ge \hp(\Gso)(D_p-\Gso) - h(\Gso)$ (by concavity of $h(G)(D_p-G)$) and  $\cpr_{\theta}(x_{\theta}) \le \cpr_{\theta}(\xso_{\theta})$. In $\Theta_3$, $\hp(\Gso)(D_p-\Gso) - h(\Gso) - \dpr_{\theta}(\xso_{\theta}) \ge 0$ and $\hp(G)(D_p-G) - h(G) \ge \hp(\Gso)(D_p-\Gso) - h(\Gso)$ and  $\cpr_{\theta}(x_{\theta}) = \cpr_{\theta}(\xso_{\theta})$.
We conclude that $\ud W (x,y) \ge 0$,  therefore $\Weq(R)$ increases with $R$. 

The case $R\in(\Rso, \Rbar)$ is handled similarly and if $R\ge \Rbar$, the equilibrium does not vary with $R$ by Proposition~\ref{prop.variation_rR}-\textit{(iii)}.

\section{Proof of \protect Proposition~\ref{prop.robustness_h}}
\label{app.proof_robustness_h_prop}

\subsection{Proof of the Proposition}

From the proof of Theorem~\ref{thm.Nash}, we know that for any mechanism $j\in\{R, T\}$, $\Geq_j$ is the fixed-point solution of \eqref{eq.fixed-point}. 
Here, we explicitly write the dependence in the mechanism, \ie for mechanism $j$, we denote by $\xresp_{j, \theta}(G)$ the individual best response~\eqref{eq.xresp}, and by $\Gresp_j(G)$ the aggregate best response~\eqref{eq.fixed-point2} to a given $G\in [0, D_p]$. 
We use the notation $r_j$ (rather than ${M^j}^{\prime}$) for the unit reward (see \eqref{eq.unit_r}). 
From the proof of Theorem~\ref{thm.SO}, we know that $\Gso$ is found as the fixed-point solution of the same equation \eqref{eq.fixed-point} with $r_j(G) = r_{SO}(G)$ defined by \eqref{eq.unit_r_SO}. Therefore, we will use the notation  $\Gso = \Geq_{SO}$ which emphasizes this similarity and helps shorten the proof's notation.

Before evaluating the variations of the equilibrium  with $\epsilon$, note that when $\epsilon=0$ (baseline), we have the same equilibrium:\vspace{-1mm}
\begin{equation}
\label{eq.sameG}
\Geq_{R}(0) = \Geq_{T}(0) = \Geq_{SO}(0),\\[-1mm]
\end{equation}
and the same unit rewards at equilibrium:\vspace{-2mm}
\begin{equation}
\label{eq.samer}
r_{R}(G) = r_{T}(G) = r_{SO}(G) \;\; \textrm{ if }\;\; G = \Geq_{R}(0).\\[-1mm]
\end{equation}

When $\epsilon\neq0$, functions $c_{\theta}$ are perturbed. %
For a given $G\in [0, D_p]$, the aggregate best response is  modified accordingly. We denote by  $\Gresp_{j, \epsilon}(G)$ the new  aggregate best response. Recall that we also denote by $\Geq_j(\epsilon)$ the new equilibrium point which is the fixed point of $\Gresp_{j, \epsilon}(\cdot)$.
The following lemma readily implies the result of Proposition~\ref{prop.robustness_h}.
\begin{lemma}
\label{lemma2}
For any $j\in\{R, T, SO\}$, we have \vspace{-2mm}
\begin{equation}
\label{eq.Geqgamma}
\Geq_j(\epsilon)  = \Geq_j(0) + \frac{J_{\epsilon}}{1+\alpha_{j}}+ o\left(\epsilon\right),\\[-1mm]
\end{equation}
where \vspace{-2mm}
\begin{equation*}
J_{\epsilon} = \Gresp_{j, \epsilon}\left(\Geq_{R}(0)\right) - \Gresp_{j, 0}\left(\Geq_{R}(0)\right)\\[-1mm]
\end{equation*}
is a first-order quantity in $\epsilon$ independent of the mechanism $j$, and\vspace{-2mm}
\begin{equation}
\label{eq.defalpha}
\alpha_{j} = -\dd{\Gresp_{j}}{G}\left(\Geq_{R}(0)\right).\\[-1mm]
\end{equation}
\end{lemma}
\begin{proof}
First note that  $\Gresp_{j, \epsilon} \left(\Geq_{R}(0)\right)$ is continuously differentiable with respect to $\epsilon$, and by assumption, the first-order term in $J_{\epsilon}$ is non-zero. Moreover, due to \eqref{eq.samer}, $\Gresp_{j, \epsilon} \left(\Geq_{R}(0)\right)$ is independent of the mechanism $j$ (see \eqref{eq.fixed-point2} and \eqref{eq.xresp}), and so is $J_{\epsilon}$.

Starting from the point $\left(\Geq_R(0), \Gresp_{j, \epsilon} \left(\Geq_{R}(0)\right)\right)$, at the first order, $\Gresp_{j, \epsilon}(G)$ decreases linearly when $G$ increases. Therefore it can be seen geometrically that it will cross again the first bisector at the new equilibrium point \vspace{-1mm}
\begin{equation}
\label{eq.Geqgamma2}
\Geq_j(\epsilon)  = \Geq_j(1) + \frac{J_{\epsilon}}{1+\alpha_{j,\epsilon}} + o\left(\epsilon\right), \\[-1mm]%
\end{equation}
where $-\alpha_{j,\epsilon}$ is the slope of the curve $\Gresp_{j, \epsilon}(G)$ at  $G=\Geq_R(0)$, \ie \vspace{-3mm}
\begin{equation}
\label{eq.defalpha2}
\alpha_{j, \epsilon} = -\dd{\Gresp_{j,\epsilon}}{G}\left(\Geq_{R}(0)\right).
\end{equation}
From~\eqref{eq.Geqgamma2}, it is easy to see that since $J_{\epsilon}$ is first-order in $\epsilon$, the first-order term in the Taylor series of $\alpha_{j, \epsilon}$ will give a second-order term in the Taylor series of $\Geq_j(\epsilon)$.  Therefore, we can restrict the series of $\alpha_{j, \epsilon}$~\eqref{eq.defalpha2} at the order zero: $\alpha_{i, \epsilon}=\alpha_{j}+o\left(1\right)$, which directly gives the desired result \eqref{eq.Geqgamma}.
\end{proof}

\subsection{Reduction of \textbf{(C1-2)} to \textbf{(C1$^\prime$-2$^\prime$)}}
\label{app.reduction}

We have \vspace{-1mm}
\begin{align*}
\dd{\Gresp_{j}}{G}(G) &= \dd{}{G} \left[  \int_{\Theta} \xresp_{j, \theta}(G)  \ud \mu ( \theta) \right], \\ 
&=   \int_{\Theta} \dd{}{G} \left[ \xresp_{j, \theta}(G) \right] \ud \mu ( \theta),   \\ 
&=   A_{j}(G) \cdot \left( \rp_j(G) - \hp(G) \right), \\[-7mm]\nonumber
\end{align*}
where\vspace{-2mm}
\begin{equation}
A_{j}(G) = \int_{\Theta_{2, j}(G)}  \left(\left( \cpr_{\theta}\right)^{-1}\right)^{\prime} (r_j(G) - h(G))  \ud \mu ( \theta), \\[-1mm]
\end{equation}
and $\Theta_{2, j}(G)$ is the subset of $\theta$'s for which $\xresp_{j, \theta}(G)\in(0, d_p)$.

At $G=\Geq_R(0)$, $\Theta_{2, j}(G)$ is independent of the mechanism $j$, so that \eqref{eq.samer} shows that $A_{j}(\Geq_R(0))$ is independent of the mechanism $j$. Denoting by $A = A_{j}(\Geq_R(0))$ the common value, we have for all $j\in\{R, T, SO\}$\vspace{-1mm}
\begin{equation*}
\alpha_j = A \cdot \left( \rp_j\left(\Geq_R(0)\right) - \hp\left(\Geq_R(0)\right) \right). \\[-1mm]
\end{equation*}

If we assume that for  $j_1\neq j_2$, $|\alpha_{j_1}-\alpha_{j_2}|$ is small, then\vspace{-1.5mm}
\begin{eqnarray}
\nonumber  \left|    \frac{1}{1+\alpha_{j_1}} - \frac{1}{1+\alpha_{j_2}}   \right| \!\!\!\!&=&\!\!\!\! \left|    \frac{\alpha_{j_2} - \alpha_{j_1}}{1+\alpha_{j_1}}\right| + o (|\alpha_{j_1}-\alpha_{j_2}|),\\
\nonumber &=&\!\!\!\! \frac{A}{1+\alpha_{j_1}} \left| \rp_{j_2} - \rp_{j_1} \right| + o (\left| \rp_{j_2} - \rp_{j_1} \right|). \\[-5mm]\nonumber
\end{eqnarray}
With this, conditions \textbf{(C1$^{\prime}$-2$^{\prime}$)} are easily deduced from \textbf{(C1-2)}.

\section{Proof of \protect Theorem~\ref{thm.robustness_h}}
\label{app.proof_robustness_h_thm}

We consider the aggregate welfare~\eqref{eq.defW} as a function of $G$: $W(G) = W(\xresp(G))$. 
We have $\dd{W}{G}(\Gso(\epsilon)) = 0$.
The result of Theorem~\ref{thm.robustness_h} is then deduced from Proposition~\ref{prop.robustness_h} using a taylor expansion around $\Gso(\epsilon)$: for $j\in\{R, T\}$,\vspace{-1mm}
\begin{equation*}
W(\Geq_j(\epsilon)) = W(\Gso(\epsilon)) + O\left( (\Geq_j(\epsilon) - \Gso(\epsilon))^2 \right).
\end{equation*}

\end{document}